\documentclass[a4paper, fleqn, oneside]{amsart}
\pdfoutput=1

\usepackage[a4paper]{geometry}

\usepackage[utf8]{inputenc}
\usepackage[T1]{fontenc}
\usepackage[colorlinks]{hyperref}
\usepackage{microtype}
\usepackage{xcolor}
\usepackage{tikz-cd}

\usepackage{enumitem}

\usepackage{comment}

\hypersetup{
  linkcolor=[rgb]{0.3,0.3,0.6},
  citecolor=[rgb]{0.2, 0.6, 0.2},
  urlcolor=[rgb]{0.6, 0.2, 0.2}
}

\usepackage{mathtools, amsthm, amsfonts, amssymb}

\usepackage{mathrsfs}
\usepackage{pbox}

\usepackage{booktabs}
\usepackage{braket}
\usepackage{commath}
\usepackage{caption}

\usepackage[capitalise]{cleveref}
\usepackage{accents}

\usepackage{tikz}
\usetikzlibrary{positioning}

\usepackage{tikz-cd}

\theoremstyle{plain}
\newtheorem{theorem}{Theorem} %
\newtheorem*{theorem*}{Theorem}
\newtheorem{proposition}[theorem]{Proposition}

\newtheorem{corollary}[theorem]{Corollary}

\theoremstyle{definition}
\newtheorem{definition}[theorem]{Definition}

\newtheorem*{problem*}{Problem}

\newtheorem{remark}[theorem]{Remark}
\newtheorem{example}[theorem]{Example}

\DeclareMathOperator{\supp}{supp}

\DeclareMathOperator{\rank}{R}
\DeclareMathOperator{\subrank}{Q}
\DeclareMathOperator{\bordersubrank}{\underline{Q}}
\DeclareMathAccent{\wtilde}{\mathord}{largesymbols}{"65}
\DeclareMathOperator{\asymprank}{\underaccent{\wtilde}{R}}
\DeclareMathOperator{\asympsubrank}{\underaccent{\wtilde}{Q}}
\DeclareMathOperator{\borderrank}{\underline{R}}

\newcommand{\Oh}{\mathcal{O}}

\newcommand{\FF}{\mathbb{F}}
\newcommand{\CC}{\mathbb{C}}

\newcommand{\NN}{\mathbb{N}}

\newcommand{\ZZ}{\mathbb{Z}}
\newcommand{\RR}{\mathbb{R}}

\DeclareMathOperator{\rk}{rk}

\newcommand{\CW}{\mathrm{CW}}

\newcommand{\combsubrank}{\subrank_\mathrm{M}}
\newcommand{\monsubrank}{\combsubrank}
\newcommand{\asympcombsubrank}{\underaccent{\wtilde}{\mathrm{Q}}_{\mathrm{M}}}
\newcommand{\monasympsubrank}{\asympcombsubrank}

\newcommand{\combdegengeq}{\unrhd_\mathrm{M}}

\newcommand{\asympgeq}{\gtrsim}

\let\leqx\leqslant

\newcommand{\doasympleqx}{%
  \hbox{\ooalign{%
    \noalign{\kern.25ex}
    $\leqslant$\cr
    \noalign{\kern1.25ex}
    \smash{$\sim$}\cr
  }}%
}

\newcommand{\doasympasympleqx}{%
  \hbox{\ooalign{%
    \noalign{\kern.25ex}
    $\leqx$\cr
    \noalign{\kern1.25ex}
    \smash{$\sim$}\cr
    \noalign{\kern0.5ex}
    \smash{$\sim$}\cr
  }}%
}

\newcommand{\doasympgeqx}{%
  \hbox{\ooalign{%
    \noalign{\kern.25ex}
    $\geqslant$\cr
    \noalign{\kern1.25ex}
    \smash{$\sim$}\cr
  }}%
}

\newcommand{\doasympdomleq}{%
  \hbox{\ooalign{%
    \noalign{\kern.25ex}
    $\preccurlyeq$\cr
    \noalign{\kern1.25ex}
    \smash{$\sim$}\cr
  }}%
}

\newcommand{\doasympasympdomleq}{%
  \hbox{\ooalign{%
    \noalign{\kern.25ex}
    $\preccurlyeq$\cr
    \noalign{\kern1.25ex}
    \smash{$\sim$}\cr
    \noalign{\kern0.5ex}
    \smash{$\sim$}\cr
  }}%
}

\DeclareMathOperator{\slicerank}{slicerank}

\newcommand{\prob}{\mathcal{P}}

\newcommand{\aspec}{\Delta}

\DeclareMathOperator{\cyc}{cyc}

\DeclareMathOperator{\irr}{\mathbf{i}}
\DeclareMathOperator{\monirr}{\mathbf{i}_{\mathrm{M}}}
\newcommand{\mongeq}{\geq_{\mathrm{M}}}
\DeclareMathOperator{\relexp}{\omega}
\DeclareMathOperator{\monrelexp}{\omega_\mathrm{M}}

\newcommand{\cw}{\mathrm{cw}}

\title[Barriers for fast matrix multiplication from irreversibility]{Barriers for fast matrix multiplication\\ from irreversibility}
\author{Matthias Christandl}
\address{Department of Mathematical Sciences, University of Copenhagen, Universitetsparken~5, 2100 Copenhagen Ø, Denmark}
\email{christandl@math.ku.dk}
\author{Péter Vrana}
\address{Department of Geometry, Budapest University of Technology and Economics, Egry József~u.~1., 1111 Budapest, Hungary\vspace{-0.5em}}
\address{MTA-BME Lend\"ulet Quantum Information Theory Research Group} 
\email{vranap@math.bme.hu}
\author{Jeroen Zuiddam}
\address{Institute for Advanced Study, 1 Einstein Drive, Princeton 08540, NJ, USA}
\email{jzuiddam@ias.edu}
\date{\today}

\begin{document}

\begin{abstract}
Determining the asymptotic algebraic complexity of matrix multiplication,
succinctly represented by the matrix multiplication exponent $\omega$,
is a central problem in algebraic complexity theory. 
The best upper bounds on $\omega$, leading to the state-of-the-art~$\omega \leq 2.37..$, have been obtained via Strassen's laser method and its generalization by Coppersmith and Winograd. 
Recent barrier results show limitations for these and related approaches to improve the upper bound on~$\omega$.

We introduce a new and more general barrier, providing stronger limitations than in previous work.
Concretely, we introduce the notion of irreversibility of a tensor, and we prove (in some precise sense) that any approach that uses an irreversible tensor in an intermediate step (e.g.,~as a starting tensor in the laser method) cannot give $\omega = 2$. In quantitative terms, we prove that the best upper bound achievable is lower bounded by twice the irreversibility of the intermediate tensor.
The quantum functionals and Strassen support functionals give (so far, the best) lower bounds on irreversibility. We provide lower bounds on the irreversibility of key intermediate tensors, including the small and big Coppersmith--Winograd tensors, that improve limitations shown in previous work. Finally, we discuss barriers on the group-theoretic approach in terms of monomial irreversibility.
\end{abstract}
\maketitle
\tableofcontents

\section{Introduction}\label{intro}
\subsection{}
Determining the asymptotic algebraic complexity of matrix multiplication is a central open problem in algebraic complexity theory. 
Several methods for constructing fast matrix multiplication algorithms have been developed, but at a high level they typically %
consist of two parts: an efficient reduction of matrix multiplication to an intermediate problem (some bilinear map, i.e.~some~3-tensor) and an efficient algorithm for the intermediate problem.
Recent results have shown \emph{barriers} for such constructions to yield fast matrix multiplication algorithms~\cite{MR3388238, MR3631613, blasiak2017groups, alman_et_al:LIPIcs:2018:8360, 8555139}.
We give a barrier, based on a new notion called irreversibility, that is more general and in some cases stronger than the barriers from previous work. %

\subsection{Matrix multiplication barriers}
The matrix multiplication exponent $\omega$ is defined as the infimum over all real numbers~$\beta$ such that any two $n\times n$ matrices can be multiplied with~$\Oh(n^\beta)$ algebraic operations, and thus~$\omega$ represents the asymptotic algebraic complexity of matrix multiplication.
The bounds $2 \leq \omega \leq 3$ hold trivially. Strassen published the first non-trivial upper bound~$\omega \leq \log_2 7$ in 1969 \cite{strassen1969gaussian}. In the decades that followed, through the development of several ingenious methods by various people, the upper bound was improved to the state-of-the-art bound $\omega \leq 2.37..$, and the pursuit to prove whether~$\omega = 2$ or $\omega >2$ has been ongoing~%
\cite{MR1056627, stothers2010complexity, MR2961552, le2014powers, cohn2003group, cohn2013fast}.  As mentioned before, these upper bound methods typically consist of a reduction of matrix multiplication to an intermediate problem and an efficient algorithm for the intermediate problem.

Ambainis, Filmus and Le Gall~\cite{MR3388238}, for the first time, proved a barrier result for some collection of such methods. Namely, they showed that a variety of methods that go via the big Coppersmith--Winograd tensor as an intermediate problem cannot give~$\omega=2$, and in fact not even $\omega \leq 2.30..$.
We call any lower bound for all upper bounds on $\omega$ that can be obtained by some method, a \emph{barrier} for that method. In general, barriers in the sense of limitations to proof methods have a long history in computational complexity theory and recognizing barriers is a natural step towards finding proof methods that do solve the problem at hand.

Next, Alman and Williams \cite{alman_et_al:LIPIcs:2018:8360, 8555139} extended the realm of barriers beyond the scope of the Ambainis \emph{et al.}~barrier, to a larger collection of methods. 
Also Blasiak, Church, Cohn, Grochow, Naslund, Sawin, and Umans~\cite{MR3631613} and Blasiak, Church, Cohn, Grochow, and Umans~\cite{blasiak2017groups} studied barriers, namely barriers for a subset of the group-theoretic method.
Both the Blasiak \emph{et al.}~and the Alman and Williams barriers rely on studying versions of asymptotic subrank of an intermediate problem.

We give a barrier that applies more generally than all previous barriers and that is in some cases stronger. Our barrier also relies on studying versions of asymptotic subrank, which together with the notion of asymptotic rank we combine into a single parameter called \emph{irreversibility}. Our barrier simplifies and generalizes previous barriers and tightly connects the barrier literature to central notions from the framework of Strassen~\cite{strassen1987relative, strassen1988asymptotic, strassen1991degeneration, christandl2017universal}.
Alman~\cite{alman2018limits} reported very similar independent results shortly after our manuscript appeared on the arXiv, which we jointly presented at the Computational Complexity Conference 2019 in New Brunswick. For all the tensors mentioned, the barriers of Alman are identical to ours. A subtle difference between the works is that Alman considers asymptotic slice rank instead of asymptotic subrank. 

\subsection{Our barrier: informal explanation}
Our barrier relies on two ideas: (i) we think of computational problems as \emph{resources} that can be reduced to one another and (ii) such reductions satisfy a \emph{triangle inequality} which limits the possible chains of reductions.
An informal explanation of our barrier is as follows. The matrix multiplication exponent~$\omega$ is the optimal \emph{rate} at which the problem of multiplying matrices can be reduced to the problem of multiplying numbers, or in other words, the optimal rate at which the problem of multiplying numbers can be transformed into the problem of multiplying matrices,
\begin{equation}
\textnormal{multiplying numbers} \,\xrightarrow{\,\omega\,}\, \textnormal{multiplying matrices}.
\end{equation}
By this we mean that for any $n$ the problem of multiplying $n \times n$ matrices can be reduced to the problem of multiplying $n^{\omega + o(1)}$ pairs of numbers (and some additions, but they do not have an influence on the complexity). We will make these notions precise later. For now, we stick to the high-level picture.
Rates of transformation naturally satisfy a triangle inequality. Therefore,  upper bounds on~$\omega$ can be obtained by combining the rate of transformation $\alpha_1$ from the problem of multiplying numbers to some intermediate problem and the rate of transformation $\alpha_2$ from the intermediate problem to the problem of multiplying matrices; this is the two-component approach alluded to earlier,
\begin{equation}\label{chain}
\textnormal{multiplying numbers} \,\xrightarrow{\,\alpha_1\,}\, \textnormal{intermediate problem} \,\xrightarrow{\,\alpha_2\,}\, \textnormal{multiplying matrices}.
\end{equation}
That is, $\alpha_1\alpha_2 \geq \omega$.
We define the \emph{irreversibility} of the intermediate tensor, roughly speaking, as the optimal rate of transformation from the problem of multiplying numbers to the intermediate problem \emph{and back to the problem of multiplying numbers}. 
Strassen \cite{strassen1988asymptotic} (see \cref{std}) showed that the transformation rate from the matrix multiplication problem to the problem of multiplying numbers is~$\tfrac12$, so %
we can extend the chain in \eqref{chain} to
\begin{multline}\label{chain2}
\textnormal{multiplying numbers} \,\xrightarrow{\,\alpha_1\,}\, \textnormal{intermediate problem} \,\xrightarrow{\,\alpha_2\,}\, \textnormal{multiplying matrices}\\ \,\xrightarrow{\,1/2\,}\, \textnormal{multiplying numbers}. 
\end{multline}
Using the triangle inequality again we see that $\alpha_1 \alpha_2 / 2$ is at least the irreversibility of the intermediate problem, and hence the irreversibility of the intermediate problem provides limitations on the upper bounds $\alpha_1 \alpha_2 \geq \omega$ that can be obtained from \eqref{chain}. This is formalized in \cref{th1}.

\subsection{Explicit numerical barriers}
To exemplify our barrier we show that the support functionals \cite{strassen1991degeneration} and quantum functionals~\cite{christandl2017universal} give (so far, the best) lower bounds on the irreversibility of the following families of intermediate problems. These intermediate problems are encoded as tensors. A tensor is a 3-dimensional array of numbers. We use $e_{i,j,k}$ to denote the tensor that is zero everywhere except for a one in coordinate $(i,j,k)$. We will discuss tensors in more details later. The families of intermediate tensors we consider are:
\begin{itemize}[leftmargin=2em]
\item the small Coppersmith--Winograd tensors 
\begin{equation*}
\cw_q = \sum_{i=1}^q e_{0,i,i} + e_{i,0,i} + e_{i,i,0}
\end{equation*}
\item the big Coppersmith--Winograd tensors 
\begin{equation*}
\CW_q = e_{0,0,q+1} + e_{0,q+1,0} + e_{q+1,0,0} + \sum_{i=1}^q e_{0,i,i} + e_{i,0,i} + e_{i,i,0}
\end{equation*}
\item the reduced polynomial multiplication tensors 
\begin{equation*}
t_n = \sum_{\substack{i,j,k=0:\\i+j = k}}^{n-1} e_{i,j,k}
\end{equation*}
\end{itemize}
Our irreversibility lower bounds lead to the following explicit barriers (\cref{cwirr} and \cref{rpm}), rounded to five decimal places.
\begin{equation*}
%
\begin{minipage}{0.3\textwidth}
\begin{tabular}{ll}
\toprule
$q$ & $\cw_q$-barrier \\
\midrule
2 & 2\\
3 & 2.02538\\
4 & 2.06244\\
5 & 2.09627\\
6 & 2.12549\\
7 & 2.15064\\
\bottomrule
\end{tabular}
\end{minipage}
\begin{minipage}{0.3\textwidth}
%
%
%
\begin{tabular}{ll}
\toprule
$q$ & $\CW_q$-barrier\\
\midrule
1 & 2.16805\\
2 & 2.17795\\
3 & 2.19146\\
4 & 2.20551\\
5 & 2.21913\\
6 & 2.23201\\
\bottomrule
\end{tabular}
\end{minipage}
\begin{minipage}{0.3\textwidth}
%
%
%
\begin{tabular}{ll}
\toprule
$n$ & $t_n$-barrier\\
\midrule
2 & 2.17795\\
3 & 2.16805\\
4 & 2.15949\\
5 & 2.15237\\
6 & 2.14641\\
7 & 2.14135\\
\bottomrule
\end{tabular}
\end{minipage}
\end{equation*}
In \cref{code} we provide code to compute the values in these tables to arbitrary precision and for arbitrary parameters $q$ and $n$.
As suggested by the values in the above tables, both the $\cw_q$-barrier and $\CW_q$-barrier increase when $q$ grows, and converge to 3 when $q$ goes to infinity. The $t_n$-barrier, on the other hand, decreases when $n$ grows, and converges to $2$ when $n$ goes to infinity. Families of tensors for which the barrier converges to 2 in the limit, like $t_n$, may be crucial in proving that $\omega = 2$.

We stress that our method can be applied to any tensor, not only the ones above. The tensors listed above play a special role in the literature and that is why we single them out.
The small and big Coppersmith--Winograd tensors come from the paper \cite{MR1056627}.
They first analyze $\cw_5$ to get an upper bound on $\omega$, and then improve it by analyzing $\CW_6$ and its Kronecker square.
All papers from then on (Stothers, Vassilevska--Williams, Le Gall) use powers of $\CW_5$ to get upper bounds on $\omega$.
The reduced polynomial multiplication tensors $t_n$ give an example of a family of tensors for which the irreducibility goes to 1 when $n$ goes to infinity and hence the barrier becomes trivial in the limit.

\subsection{Comparison and other applications}

Compared to~\cite{MR3388238}, whose barriers apply to the laser method, our barriers are valid for a larger class of approaches (and naturally we obtain lower barriers). 
Compared to~\cite{8555139}, whose barriers apply to what we call monomial degeneration, our barriers are valid for a larger class of approaches but our barriers are also higher.
As a variation on our barrier we introduce a \emph{monomial} version.
Compared to~\cite{MR3631613} and~\cite{blasiak2017groups} our monomial barriers are valid for a class of approaches that includes their simultaneous triple product property (STPP) approach, and thus we provide a uniform view on the barriers that have appeared in the literature.
We have not tried to optimize the barriers that we obtain, but focus instead on introducing the barrier~itself. The barrier in~\cite{alman2018limits} is very similar to ours, except for using asymptotic slice rank instead of asymptotic subrank. 

It will become clear to the reader during the development of our ideas that they not only apply to the problem of fast matrix multiplication, but extend to give barriers for the more general problem of constructing fast rectangular matrix multiplication algorithms (see the follow-up results on barriers for rectangular matrix multiplication in \cite{christ2020barriers}) or even transformations between arbitrary powers of tensors. Such transformations may represent, for example, asymptotic SLOCC (stochastic local operations and classical communication) reductions among multipartite quantum states \cite{bennett2000exact, MR1804183, MR1910235,MR2515619}.

\medskip
We define irreversibility in \cref{sec:irr}. In \cref{sec:barriers} we discuss how irreversibility implies a barrier. In \cref{sec:methods} we discuss methods to analyze irreversibility. Finally, in \cref{sec:ex} we exhibit explicit irreversibility barriers.

\section{Irreversibility of tensors}\label{sec:irr}
We begin by introducing some standard terminology and results regarding tensors and matrix multiplication. Then we discuss a useful notion called the relative exponent of two tensors and we define the irreversibility of a tensor. After that we introduce the monomial versions of these ideas. Finally, we discuss what values the irreducibility can take.

\subsection{Tensors and matrix multiplication}\label{std}  
We begin with some standard terminology and results regarding tensors and matrix multiplication. Standard references for this material are~\cite{burgisser1997algebraic} and \cite{blaser2013fast}.
Let $\FF$ be a field. Everything we discuss works over all fields, except when mentioned otherwise.
An \emph{$n_1 \times n_2 \times n_3$ tensor} (over $\FF$) is a three-dimensional array 
\[
t = (t_{i_1, i_2, i_3})_{i_1 \in [n_1], i_2 \in [n_1], i_3 \in [n_1]}
\]
of field elements $t_{i_1, i_2, i_3} \in \FF$. We denote the set of  all $n_1 \times n_2 \times n_3$ tensors by $\FF^{n_1 \times n_2 \times n_3}$.
The \emph{support of $t$} is defined as the set
\[
\supp(t) \coloneqq \{(i_1, i_2, i_3) \in [n_1] \times [n_2] \times [n_3] : t_{i_1, i_2, i_3} \neq 0\}.
\]
We use $e_{i,j,k}$ to denote the tensor that is zero everywhere except for a one in coordinate $(i,j,k)$.
For any $m_i \times n_i$ matrices $A_i$ over $\FF$ we define the tensor $(A_1, A_2, A_3)\cdot t \in \FF^{m_1 \times m_2 \times m_3}$ by 
\[
((A_1, A_2, A_3)\cdot t)_{j_1, j_2, j_3} = \sum_{i_1, i_2, i_3} (A_1)_{j_1, i_1} (A_2)_{j_2, i_2} (A_3)_{j_3, i_3} \, t_{i_1, i_2, i_3}.
\]
For $t \in \FF^{n_1 \times n_2 \times n_3}$ and~$s \in \FF^{m_1 \times m_2 \times m_3}$
we write~$t \geq s$ and say that \emph{$t$ restricts to~$s$} if there are matrices $A_i$ such that $(A_1, A_2, A_3)\cdot t = s$. We say $t$ and $s$ are \emph{equivalent} if $t \geq s$ and $s \geq t$. We say $t$ and $s$ are \emph{isomorphic} and write $s \cong t$ if $s = (A_1, A_2, A_3)\cdot t$ for invertible matrices $A_i$. We will often simplify the notation and simply write $s = t$ when $s$ and $t$ are isomorphic.

For $n \in \NN$ we define the \emph{rank-$n$ unit tensor} $\langle n \rangle \coloneqq \sum_{i=1}^n e_{i,i,i} \in \FF^{n \times n \times n}$.
A \emph{simple tensor} is a tensor that is an outer product $\sum_{i,j,k} u_i v_j w_k\, e_{i,j,k}$ of three vectors $u$, $v$ and $w$.
The \emph{tensor rank of~$t$} is defined as the smallest $r$ such that $t$ can be written as a sum of $r$ simple tensors, that is, $t = \sum_{\ell=1}^r \sum_{i,j,k} u^{(\ell)}_i v^{(\ell)}_j w^{(\ell)}_k\, e_{i,j,k}$ for vectors $u^{(\ell)}$, $v^{(\ell)}$ and $w^{(\ell)}$.
It is crucial for us, but not hard to see, that we can phrase tensor rank in terms of the restriction preorder and diagonal tensors via~$\rank(t) = \min \{n \in \NN : t \leq \langle n\rangle\}$. Indeed, if $t  = (A_1, A_2, A_3) \cdot \langle r\rangle$, then the columns of the matrices $A_1$, $A_2$ and $A_3$ provide the vectors $u^{(\ell)}$, $v^{(\ell)}$ and $w^{(\ell)}$ for a tensor rank upper bound, and in the other direction such vectors $u^{(\ell)}$, $v^{(\ell)}$ and $w^{(\ell)}$ provide the matrices $A_1$, $A_2$ and $A_3$.
What is nice about this definition is that it naturally gives rise to the following tensor parameter, which is central for our barriers:
The \emph{subrank of~$t$} is defined as~$\subrank(t) \coloneqq \max \{n \in \NN : \langle n \rangle \leq t\}$.

Extending the definition of Kronecker product for matrices, we define the \emph{Kronecker product} $s \otimes t \in \FF^{(m_1n_2) \times (m_2n_2) \times (m_3n_3)}$ for tensors $s \in \FF^{m_1 \times m_2 \times m_3}$ and $t \in \FF^{n_1 \times n_2 \times n_3}$ as the tensor that has coefficients given by all pairwise products of coefficients of $s$ and $t$. Extending the definition of direct sum for matrices, we define the \emph{direct sum} $s \oplus t \in \FF^{(m_1 + n_2) \times (m_2 + n_2) \times (m_3 + n_3)}$ as the tensor obtained by placing $s$ and $t$ as blocks in a block-diagonal tensor.
In long expressions we will often write $t^n$ to denote $t^{\otimes n}$.
The product $\otimes$, the sum $\oplus$ and the diagonal tensors $\langle n\rangle$ behave well with respect to each other. One fact we will use frequently is that the direct sum of $n$ copies of a tensor $t$ is isomorphic to the Kronecker product of the diagonal tensor $\langle n\rangle$ and $t$.

For any tensor $t \in \FF^{n_1 \times n_2 \times n_3}$ we define linear maps $t_1 : \FF^{n_1} \to  \FF^{n_2 \times n_3} : v \mapsto \sum_{i,j,k} v_i t_{i,j,k} e_{j,k}$, 
$t_2 : \FF^{n_2} \to  \FF^{n_1 \times n_3} : v \mapsto \sum_{i,j,k} v_j t_{i,j,k} e_{i,k}$
and
$t_3 : \FF^{n_3} \to  \FF^{n_1 \times n_2} : v \mapsto \sum_{i,j,k} v_k t_{i,j,k} e_{i,j}$. These are called the \emph{flattenings} of $t$. We call each function $t \mapsto \rk(t_i)$ a \emph{flattening rank} of $t$. A basic observation and tool is that each flattening rank is multiplicative under the Kronecker product on tensors, additive under the direct sum on tensors, and monotone under the preorder $\leq$ on tensors.

The \emph{asymptotic rank of~$t$} is defined as
\begin{equation}
\asymprank(t) \coloneqq \lim_{n \to \infty} \rank(t^{\otimes n})^{1/n} = \inf_n \rank(t^{\otimes n})^{1/n}
\end{equation}
and the \emph{asymptotic subrank of~$t$} is defined as
\begin{equation}
\asympsubrank(t) \coloneqq \lim_{n \to \infty} \subrank(t^{\otimes n})^{1/n} = \sup_n \subrank(t^{\otimes n})^{1/n}.
\end{equation} 
We need to argue that the limits in the definitions exist.
If $t \leq \langle n\rangle$ and $s \leq \langle m\rangle$, then we also have $t \otimes s \leq \langle nm\rangle$. It follows from this that tensor rank is sub-multiplicative under $\otimes$. Similarly, subrank is super-multiplicative under $\otimes$. 
Therefore, by Fekete's lemma,\footnote{Fekete's lemma is as follows. Let $r_1, r_2, r_3, \ldots \in \RR_{\geq 0}$ satisfy $r_{n+m} \leq r_n + r_m$ for all $n,m$. Then it holds that $\lim_{n\to\infty} r_n/n = \inf_n r_n/n$. See, e.g.,~\cite[No.~98]{MR1492447}.} the above limits exist and equal the respective infimum and supremum. We note that asymptotic rank and asymptotic subrank may equivalently be defined using \emph{border} rank and \emph{border} subrank, respectively, which are the approximative versions of rank and subrank (e.g., see \cite{blaser2013fast}).

For~$a,b,c \in \NN_{\geq1}$ the \emph{matrix multiplication tensor} $\langle a,b,c\rangle$ is defined as
\begin{equation}
\langle a,b,c \rangle \coloneqq \sum_{i=1}^a \sum_{j=1}^b \sum_{k=1}^c e_{(i,j),\, (j,k),\, (k,i)} \in \FF^{ab \times bc \times ca}.
\end{equation}
To see how $\langle n,n,n\rangle$ encodes multiplication of $n \times n$ matrices, let  $E_{i,j}$ be the $n\times n$ matrix that is all-zero except for a 1 at coordinate $(i,j)$, so that the $E_{i,j}$ form the standard basis of the vector space of $n\times n$ matrices. Note that $E_{i,j} E_{j,k} = E_{i,k}$ and that $\sum_{i=1}^n \sum_{j=1}^n \sum_{k=1}^n A_{(i,j)} B_{(j,k)} E_{i,k} = AB$.
Thus suitable contractions of $\langle n,n,n\rangle$ with the matrices $A$ and $B$ produces the product $AB$.
It is a standard result (\cite[Section 4]{blaser2013fast}) that the tensor rank of $\langle n,n,n\rangle$ equals, up to a constant factor, the arithmetic complexity of multiplying two $n\times n$ matrices. The argument is roughly as follows. Tensor rank upper bounds lead directly to matrix multiplication algorithms by the aforementioned contractions. On the other hand, any arithmetic matrix multiplication algorithm can be made into an arithmetic circuit with an addition layer followed by a bilinear multiplication layer followed by an addition layer, in such a way that the number of bilinear multiplication nodes is at most twice the number of multiplications in the arithmetic algorithm. The blow-up by a factor of two is because we only allow \emph{bilinear multiplications} in our arithmetic circuit. This representation leads directly to a tensor rank upper bound.
As we said before, the \emph{matrix multiplication exponent}~$\omega$ is defined as the infimum over all $\beta$ for which the arithmetic complexity of multiplying two~$n\times n$ matrices is at most $\Oh(n^\beta)$. Matrix multiplication tensors are multiplicative in the sense that $\langle a,b,c\rangle \otimes \langle d,e,f\rangle$ is isomorphic to $\langle ad, be, cf\rangle$.
From the multiplicativity of the matrix multiplication tensors it is easy to derive that $\omega$ is characterized by the asymptotic rank of $\langle 2,2,2\rangle$ as follows.

\begin{proposition}[{see, e.g., \cite[Theorem 5.2]{blaser2013fast}}]\label{prop:omega}
$\omega = \log_2 \asymprank(\langle 2,2,2\rangle)$.
\end{proposition}

The difficulty of determining the asymptotic rank of $\langle2,2,2\rangle$ is to be contrasted with the situation for the asymptotic subrank; to put it in Strassen's words:
 \emph{Unlike the cynic, who according to Oscar Wilde knows the price of everything and the value of nothing, we can determine the asymptotic value of $\langle h,h,h\rangle$ precisely}.
\begin{proposition}[\cite{strassen1988asymptotic}]\label{prop:wecan}
$2 = \log_2 \asympsubrank(\langle 2,2,2\rangle)$.
More precisely, for any $h \in \NN$ it holds that
\begin{equation}\label{wecan}
\asympsubrank(\langle h,h,h\rangle) = h^2.
\end{equation}
\end{proposition}
For completeness of the paper we give a proof of \cref{prop:wecan}. The proof we provide here is based on Salem--Spencer sets and is slightly different from the proof in \cite{strassen1988asymptotic}. However, the two proofs become very similar when the construction of Salem--Spencer sets is unrolled.
\begin{proof}
Since each flattening rank of the diagonal tensor $I_n$ equals $n$ and since the flattening rank of $\langle h,h,h\rangle$ equals $h^2$ we find by multiplicativity and monotonicity of the flattening rank that~$\asympsubrank(\langle h,h,h\rangle) \leq h^2$. We must prove the matching lower bound.
Let $p$ be a large prime. There is a subset $A \subseteq \ZZ/p\ZZ$ (a \emph{Salem--Spencer set} in $\ZZ/p\ZZ$) that is free of nontrivial three-term arithmetic progressions and that has size~$\abs[0]{A} = p^{1-o(1)}$ \cite{MR7405,MR18694}.
The support of the tensor $\langle p,p,p\rangle$ is the set of triples
\[
E = \{((i,j), (j,k), (k,i)) : i,j,k \in [p]\}.
\]
We think of $E$ as a 3-partite 3-uniform hypergraph $E \subseteq V_1 \times V_2 \times V_3$ with vertex sets $V_\ell = [p]\times [p]$.
Consider the sub-hypergraph $F \subseteq E$ induced by the vertex sets
\begin{gather*}
W_1 = \{(x, x+a) \in V_1 : x \in [p], a \in A\}\\
W_2 = \{(x, x+a) \in V_2 : x \in [p], a \in A\}\\
W_3 = \{(x+2a, x) \in V_3: x \in [p], a \in A\}.
\end{gather*}
That is, $F = E \cap (W_1 \times W_2 \times W_3)$. Then 
\[
F = \{((i,j), (j,k), (k,i)) : (i,j,k) = (x, x+a, x+2a) \textnormal{ for some } x \in [p] \textnormal{ and } a \in A\}
\]
Indeed, a priori, $((i,j), (j,k), (k,i)) \in F$ if and only if
\[
(i,j,k) = (x, x+a, x+a+b = x + 2c)
\]
for some $x \in [p]$ and $a,b,c \in A$. Then $a, (a+b)/2, b$ is a three-term arithmetic progression in~$A$, and thus $a = b$ and the claim follows.
We see that the set $F$ has $p \abs[0]{A} = p^{2 - o(1)}$ elements.
We also see that for any triple $((i,j), (j,k), (k,i))$ in $F$, the first vertex~$(i,j)$ determines the second vertex~$(j,k)$ and the third vertex~$(k,i)$, and similarly the second vertex determines the first and third vertex and the third vertex determines the first and second vertex.
In other words, the edges in $F$ are disjoint in all coordinates, that is, they form a matching.
Restrict the tensor~$\langle p,p,p\rangle \in \FF^{V_1 \times V_2 \times V_3}$ to $\FF^{W_1 \times W_2 \times W_3}$ by restricting the coordinates in each direction from~$V_\ell$ to~$W_\ell$. The support of the resulting tensor is the matching $F$. After a suitable permutation of each $W_\ell$, this tensor becomes a diagonal tensor of size $|F| = p^{2 - o(1)}$.
We conclude that for every prime~$p$ we have $\subrank(\langle p,p,p\rangle) \geq p^{2 - o(1)}$. By a simple tensor product argument this implies the claim for all $h \in \NN$. Namely, for every $h \in \NN$ we take a prime $p$ such that $h^n$ is not much bigger than an integer power of $p$, that is, $h^n \geq p^m \geq h^n/p$. Then  $\langle h^n,h^n,h^n\rangle \geq \langle p^m, p^m, p^m\rangle$ and we find that $\subrank(\langle h^n, h^n, h^n\rangle) \geq h^{2n - o(n)}$.
\end{proof}

\begin{remark}\label{rem:strproof}
We remarked before that Strassen's proof of \cref{prop:wecan} in \cite{strassen1988asymptotic} takes a slightly different route. Namely, it makes use of an approximative relaxation of subrank called \emph{border subrank}, which is denoted by $\bordersubrank$. He shows that $\bordersubrank(\langle m,m,m\rangle) \geq \lceil\tfrac34m^2 \rceil$.  (This lower bound, and all Strassen's lower bounds for $\bordersubrank(\langle a,b,c\rangle)$, are in fact tight~\cite{kopparty2020geometric}.) Strassen then relates the border subrank to the asymptotic subrank via a simple polynomial interpolation argument to get that~$\asympsubrank(\langle m,m,m\rangle) = m^2$. 
\end{remark}

\subsection{Relative exponent}
For a clean exposition of our barrier we will use the notion of relative exponent, which we will define in this section. 
This notion is inspired by the notion of rate from information theory and alternatively can be seen as a versatile version of the notion of the asymptotic preorder for tensors of Strassen\footnote{The asymptotic preorder $\asympgeq$ is defined as follows: $s \asympgeq t$ if $s^{\otimes n + o(n)} \geq t^{\otimes n}$ when $n \to \infty$.}. In the context of tensors, the relative exponent previously appeared in \cite{yu2014obtaining},\cite{vrana2013asymptotic} and \cite{christandl2016asymptotic}.

To avoid technicalities, we will from now on, without further mentioning, only consider tensors that are not of tensor rank one or zero. This assumption is explicitly used in the proof of \cref{basic}(i) and its applications.

\begin{definition}
For two tensors $t \in \FF^{n_1 \times n_2 \times n_3}$ and $s \in \FF^{m_1 \times m_2 \times m_3}$ we define the \emph{relative exponent from $t$ to $s$}~as
\begin{align}
\relexp(t, s) &\coloneqq \lim_{n\to \infty} \tfrac1n \min\{ m \in \NN_{\geq 1} : t^{\otimes m} \geq s^{\otimes n} \}\\
&= \sup_n \tfrac1n \min\{ m \in \NN_{\geq1} : t^{\otimes m} \geq s^{\otimes n} \}.
\end{align}
\end{definition}
The limit is a supremum by Fekete's lemma.
Let us briefly relate the relative exponent to the basic notions and results stated earlier.
The reader verifies directly that the identities
\begin{gather}
\relexp(\langle 2\rangle, t) = \log_2 \asymprank(t)\\
\relexp(t, \langle 2\rangle) = 1/(\log_2 \asympsubrank(t))
\end{gather}
hold by approximating integers by powers of 2, or see \cite[Proposition 1.1.16]{christandl2016asymptotic} for a proof.
From the fact that $\omega \coloneqq \log_2 \asymprank(\langle 2,2,2\rangle)$ (\cref{prop:omega}) it follows that
\begin{equation}
\relexp(\langle 2\rangle, \langle2,2,2\rangle) = \omega.
\end{equation}
We know from \eqref{wecan} that $\asympsubrank(\langle 2,2,2\rangle) = 4$ and so
\begin{equation}\label{subexp}
\relexp(\langle 2,2,2 \rangle, \langle 2\rangle) = \tfrac12.
\end{equation}
The relative exponent has the following two basic properties.
\begin{proposition}\label{basic} Let $s$, $t$ and $u$ be tensors.
\begin{enumerate}[label=\upshape(\roman*)]
\item $\relexp(t, t) = 1$.
\item $\relexp(s, t)\relexp(t, u) \geq \relexp(s,u)$ \quad\textnormal{(triangle inequality)}.
\end{enumerate}
\end{proposition}
\begin{proof}
(i) Clearly $\omega(t,t) \leq 1$ by reflexivity of the restriction preorder $\geq$. 
Since $t$ is not of rank 1, we can flatten $t$ into a matrix in one of the three directions, so that the matrix rank is some number $r > 1$. It follows from multiplicativity of matrix rank, that if $t^{\otimes m} \geq t^{\otimes n}$, then $r^m \geq r^n$ and so $m \geq n$. This implies the claim.
(ii) follows from the transitivity of the restriction preorder~$\geq$.
\end{proof}

\subsection{Irreversibility.} 
Our barrier framework relies crucially on the irreversibility of a tensor, a new notion that we define now.

\begin{definition}
We define the \emph{irreversibility} of a tensor $t$ as the product of the relative exponent from $\langle 2\rangle$ to $t$ and the relative exponent from $t$ to $\langle 2\rangle$, i.e.
\begin{equation}
\irr(t) \coloneqq \relexp(\langle 2\rangle, t) \relexp(t, \langle2\rangle).
\end{equation}
Thus $\irr(t)$ measures the extent to which the asymptotic conversion from $\langle 2\rangle$ to $t$ is irreversible, explaining the name.
Equivalently, the irreversibility is the ratio of the logarithms of the asymptotic rank and the asymptotic subrank, i.e.
\begin{equation}
\irr(t) = \frac{\log_2 \asymprank(t)}{\log_2 \asympsubrank(t)}.
\end{equation}
\end{definition}

\begin{proposition}\label{biggerthanone} For any tensor $t$ it holds that
\begin{equation}
\irr(t) \geq 1.
\end{equation}
\end{proposition}
\begin{proof}
From the basic properties of the relative exponent (\cref{basic}) it follows directly that the inequality~$\irr(t) = \relexp(\langle2\rangle, t) \relexp(t, \langle 2\rangle) \geq \relexp(\langle2\rangle, \langle2\rangle) = 1$ holds.
\end{proof}

\begin{definition} We call a tensor $t$ \emph{reversible} if $\irr(t) = 1$ and \emph{irreversible} otherwise (in which case it holds that $\irr(t) > 1$ by \cref{biggerthanone}).
\end{definition}
For example, for any $n \in \NN$ the diagonal tensor $\langle n\rangle = \sum_{i=1}^n e_{i,i,i}$ is reversible. In fact, every reversible tensor $t$, \emph{that we know of}, is equivalent to $\langle n\rangle$ for some $n$, in the sense that $\langle n\rangle \leq t \leq \langle n\rangle$. In other words, we do not know whether every reversible tensor $t$ is equivalent to $\langle n\rangle$ for some $n$.

This question is closely related to whether the matrix multiplication exponent $\omega$ equals 2.
For the matrix multiplication tensor $\langle 2,2,2\rangle$ we have that $2 \irr(\langle2,2,2\rangle) = \omega$ (using \eqref{subexp}). Thus~$\omega = 2$ if and only if $\langle 2,2,2 \rangle$ is reversible. (In fact, for any $n \in \NN$ it holds that $\omega = 2$ if and only if~$\langle n,n,n\rangle$ is reversible.)
As we will see in \cref{sec:barriers}, this is ultimately the source of our barrier. %

Irreversible tensors do in fact exist.
For example, the tensor $W = e_{0,0,1} + e_{0,1,0} + e_{1,0,0}$ is irreversible. Namely, it is known that $\log_2 \asymprank(W) = 1$ and that $\log_2 \asympsubrank(W) = h(1/3) = 0.918..$ \cite[Theorem 6.7]{strassen1991degeneration}, so~$\irr(W) = 1.088.. > 1$. 
In \cref{sec:ex} we will compute lower bounds on the irreversibility of the small and big Coppersmith--Winograd tensors (which play a crucial role in the best upper bounds on $\omega$).

\subsection{Monomial relative exponent and monomial irreversibility}\label{monomial}
The following restricted version of relative exponent and irreversibility will be relevant.
A \emph{generalized sub-permutation matrix} is a matrix for which every row and every column has at most one nonzero coefficient.
For two tensors $t \in \FF^{n_1 \times n_2 \times n_3}$ and~$s \in \FF^{m_1 \times m_2 \times m_3}$ we write~$t \mongeq s$ and say $t$ \emph{monomially restricts} to~$s$ if there are linear maps $A_i : \FF^{n_i} \to \FF^{m_i}$, the corresponding matrices of which are generalized sub-permutation matrices in the standard basis, such that $(A_1, A_2, A_3)\cdot t = s$ \cite[Section~6]{strassen1987relative}. 
Thus $s$ is a monomial restriction of $t$ if we can obtain $s$ from $t$ by rescaling (also to zero) and permuting the slices of $t$. Monomial restriction plays a central role in the existing upper bounds on $\omega$.

Replacing the preorder $\geq$ by $\mongeq$ in \cref{sec:irr} gives the notions of monomial subrank~$\monsubrank$, monomial asymptotic subrank~$\monasympsubrank$ and monomial relative exponent $\monrelexp$. 
(For simplicity we will use monomial restriction here, but our results will also hold with $\mongeq$ replaced by monomial degeneration $\combdegengeq$ defined in \cite[Section 6]{strassen1987relative}. Monomial degeneration is the approximative version of monomial restriction. It is reflexive and transitive like monomial restriction.)
Note that the notions $\monsubrank$ and~$\monasympsubrank$ only depend on the support of the tensor, and not on the particular values of the nonzero coefficients.
We define the monomial irreversibility $\monirr(t)$ of $t$ as the product of the (normal) relative exponent from $\langle2\rangle$ to $t$ and the monomial relative exponent from $t$ to $\langle 2\rangle$,
\begin{equation}
\monirr(t) \coloneqq \relexp(\langle2\rangle, t)\monrelexp(t, \langle2\rangle).
\end{equation}
Equivalently, we have
\begin{equation}
\monirr(t) = \frac{\log_2 \asymprank(t)}{\log_2 \monasympsubrank(t)}.
\end{equation}
We stress that the monomial irreversibility uses the monomial asymptotic subrank, but the \emph{regular} asymptotic rank. This is in fact what is used in practice.
We note that monomial (asymptotic) rank is a useless concept, since any monomial restriction of a unit tensor is again a unit tensor.
In other words, using the monomial asymptotic rank is too restrictive and gives bad algorithms from the start. In particular, monomial irreversibility depends on the tensor and not only on its support.
\begin{proposition}\label{basicmonx} Let $s$, $t$ and $u$ be tensors.
\begin{enumerate}[label=\upshape(\roman*)]
\item $\monrelexp(t, t) = 1$.
\item $\monrelexp(s, t)\monrelexp(t, u) \geq \monrelexp(s,u)$ \quad\textnormal{(triangle inequality)}.
\item $\monrelexp(s, t) \geq \relexp(s,t)$.
\item $\monirr(t) \geq \irr(t)$.
\end{enumerate}
\end{proposition}
\begin{proof}
The first two statements follow from reflexivity and transitivity of the monomial restriction preorder $\mongeq$, just like in the proof of \cref{basic}.
The other two statements follow from the fact that $s \mongeq t$ implies $s \geq t$.
\end{proof}

\begin{definition}
We call a tensor $t$ \emph{monomially reversible} if $\monirr(t) = 1$ and \emph{monomially irreversible} otherwise (in which case it holds that $\monirr(t) > 1$ by \cref{basicmonx} (iv) and \cref{biggerthanone}).
\end{definition}
There exist tensors that are reversible and monomially irreversible. We give an example. 

\begin{example}\label{ste}
For any finite group $G$, we define the tensor $\langle G\rangle_\FF \in \FF^{G \times G \times G}$ by 
\[
\langle G\rangle_\FF \coloneqq \sum_{g,h \in G} e_{g,h,gh}.
\] 
The tensor $\langle G\rangle_\FF$ is often referred to as the \emph{structure tensor} of the group algebra $\FF[G]$. The group algebra $\FF[G]$ is the vector space $\FF^G$ that has a bilinear multiplication operation defined on it by setting $e_g \cdot e_h = e_{gh}$ for the basis elements $\{e_g : g \in G\}$ and extending bilinearly to~$\FF[G]$. Thus~$\langle G\rangle_\FF$ precisely encodes this multiplication.
We stress that in the notation~$\langle G\rangle_\FF$ the parameter~$G$ determines the support and the parameter~$\FF$ denotes over which field we consider the tensor. Even though all coefficients of~$\langle G\rangle_\FF$ are 0 or~1, the choice of base field~$\FF$ determines what linear maps are allowed in the restriction preorder $\geq$ and thus the notion of equivalence and the notion of isomorphism. Crucially, for the monomial restriction preorder $\mongeq$ applied to tensors with coefficients that are only 0 or 1, the choice of base field does not make a difference.

We consider the structure tensor
\begin{equation}
\langle \ZZ/3\ZZ \rangle_\FF = e_{0,0,0} + e_{0,1,1} + e_{1,0,1} + e_{2,0,2} + e_{0,2,2} + e_{1,1,2} + e_{1,2,0} + e_{2,1,0} + e_{2,2,1}.
\end{equation}
First of all, $\langle \ZZ/3\ZZ \rangle_\CC$ is isomorphic to the diagonal tensor $\langle 3 \rangle$. In general, for any finite abelian group $A$, the structure tensor $\langle A\rangle_\CC$ is isomorphic to the diagonal tensor $\langle|A|\rangle$, using the Fourier transform on $A$ to diagonalize the multiplication (see, e.g., \cite{MR3631613} or \cite{christandl2017universal} for a discussion of this).
Therefore, $\asymprank(\langle \ZZ/3\ZZ \rangle_\CC) = 3$ and $\asympsubrank(\langle \ZZ/3\ZZ \rangle_\CC) = 3$ and hence $\langle \ZZ/3\ZZ\rangle_\CC$ is reversible.

In fact, for general $\FF$, we have $\asymprank(\langle \ZZ/3\ZZ \rangle_\FF) = 3$ because there is a monomial degeneration from the restricted polynomial multiplication tensor to $\langle \ZZ/3\ZZ \rangle_\FF$, see \cite[Theorem 4.16]{christandl2017universal}.  On the other hand, it is known that $\monasympsubrank(\langle \ZZ/3\ZZ \rangle_\FF) = 2.755..$. This is proven in \cite{MR3583358, tao}, see also~\cite{christandl2017universal} for the connection to \cite{strassen1991degeneration}. It follows for any field $\FF$ that 
$\monirr(\langle \ZZ/3\ZZ \rangle_\FF) = 1.08..$  and that $\langle \ZZ/3\ZZ\rangle_\FF$ is monomially irreversible.
\end{example}

Regarding matrix multiplication, Strassen's construction for \eqref{subexp} (see also our proof of \cref{prop:wecan}) in fact shows that
\begin{equation}
\monrelexp(\langle2,2,2\rangle, \langle2\rangle) = \tfrac12.
\end{equation}

\subsection{Upper bounds on irreversibility}\label{balanced}
We finish this section by discussing the possible values that the irreversibility can take, as a first step towards a systematic understanding of where we may find reversible tensors or almost reversible tensors. We will see that, surprisingly at first sight, the matrix multiplication exponent $\omega$ provides bounds on the irreversibility of arbitrary tensors. 

First of all, if~$\asympsubrank(t) \leq 1$, then~$\irr(t) = \infty$. Thus, a priori, the irreversibility of a tensor could be any number in $[1, \infty]$.
Recall that we defined the flattenings of any tensor $t \in \FF^{n_1 \times n_2 \times n_3}$ as the linear maps $t_1 : \FF^{n_1} \to  \FF^{n_2 \times n_3}$, $t_2 : \FF^{n_2} \to  \FF^{n_1 \times n_3}$ and $t_3 : \FF^{n_3} \to  \FF^{n_1 \times n_2}$ that are naturally associated to $t$.
For any tensor $t$ that has a flattening with rank (as a linear map) at most~1 it holds that $\asympsubrank(t) \leq 1$. For example, $\asympsubrank(\langle n,1,1\rangle) = 1$.  If none of the flattening ranks is at most~1, then the asymptotic subrank is at least some absolute constant $c > 1$, as follows from~\cite[Lemma~3.7]{strassen1988asymptotic}. Thus in that case the irreversibility is finite. On the other hand, for any $t \in \FF^{n \times n \times n}$ the upper bound $\asymprank(t) \leq \rank(t) \leq n^2$ holds. We conclude that, for $t \in \FF^{n \times n \times n}$,
if any flattening rank is at most 1, then $\irr(t) = \infty$, and otherwise
\[
1 \leq \irr(t) \leq d \log_2 n
\]
for some absolute constant $d$. This gives us an idea of the order of magnitude of the irreversibility.
We will refine this upper bound in the rest of this discussion.

As the first step towards gaining more control on the possible values of the irreversibility, we note that for asymptotic rank we have the following general upper bound.

\begin{proposition}[{\cite[Proposition~3.6]{strassen1988asymptotic}}]\label{arub}
For $t \in \FF^{n \times n \times n}$ it holds that $\asymprank(t) \leq  n^{2\omega/3}$.
\end{proposition}
\begin{proof}
We give a  sketch of the argument. For any tensor $t \in \FF^{n \times n \times n}$ it holds that $t \leq \langle n,n,1\rangle$ and $t \leq \langle n,1,n\rangle$ and $t \leq \langle 1, n, n\rangle$. By multiplying these inequalities, it follows that $t^{\otimes 3} \leq \langle n^2, n^2, n^2\rangle$. It is not hard to see that therefore $\asymprank(t) \leq \asymprank(\langle n,n,n\rangle)^{2/3}  = n^{2\omega /3}$. 
\end{proof}

It is possible that $\omega = 2$, in which case it follows from \cref{arub} that $\asymprank(t) \leq n^{4/3}$ for every $t \in \FF^{n \times n \times n}$.
More generally, it is possible that for every tensor $t \in \FF^{n \times n \times n}$ we have the highly non-trivial upper bound $\asymprank(t) \leq n$, see \cite[Problem 15.5]{burgisser1997algebraic}. Strassen conjectured this to be true for a subset of all tensors called \emph{tight} tensors \cite[Conjecture 5.3]{MR1341854}.

\cref{arub} gives us some control over the possible values of the irreversibility. As the next step in that direction we must understand how small the asymptotic subrank can be.
First, we give an example of a tensor for which the asymptotic subrank and the asymptotic rank are relatively far apart.

\begin{example}\label{ex:notbalanced}
Let $t \in \FF^{n \times n \times n}$ be the tensor $t = e_{1, n, 1} + e_{n, 1, 1} + e_{2, n, 2} + e_{n, 2, 2} + \cdots + e_{n, n, n}$. Then for $n\geq 5$ it is known~\cite[Equation~(6.19)]{strassen1991degeneration} that~$\asympsubrank(t) = 2\sqrt{n-1}$ while $\asymprank(t) = n$ and so we have that $\irr(t)$ goes to 2 (from below) when $n$ goes to infinity.
\end{example}

For the rest of the discussion we require our tensors to satisfy a property called \emph{balanced}. This property was introduced in~\cite[page 121]{strassen1988asymptotic}.
Not all tensors are balanced, but in any tensor space $\FF^{n \times n \times n}$ over an algebraically closed field $\FF$, being balanced is a generic condition in the sense of algebraic geometry. In particular,~almost all elements in $\FF^{n \times n \times n}$ are balanced.
We define a tensor $t \in \FF^{n_1 \times n_2 \times n_3}$ to be \emph{balanced} if the three flattenings 
that we defined before are full-rank as linear maps and for each $i \in [3]$ there is an element $v \in \FF^{n}$ such that the element~$t_i(v)$ has full rank (as a matrix). Balanced tensors are called \emph{1-generic tensors} in~\cite{LANDSBERG2017333}. An example of a non-balanced tensor is the tensor in \cref{ex:notbalanced}. An example of a balanced tensor is the matrix multiplication tensor $\langle n,n,n\rangle$. 
\begin{proposition}\label{prop:ub3}
Let $t \in \FF^{n \times n \times n}$ be balanced. Then
\[
1 \leq \irr(t) \leq \omega.
\]
\end{proposition}
\begin{proof}
The claim follows from the following two ingredients. The first ingredient is the  general upper bound $\asymprank(t) \leq n^{2 \omega/3}$ from \cref{arub}.
The second ingredient is the lower bound
$\asympsubrank(t) \geq n^{2/3}$ \cite[Proposition~3.6]{strassen1988asymptotic} for balanced tensors. We give a sketch of the argument.
From the balancedness assumption it follows that $\langle n, 1, 1\rangle \leq t$ and $\langle 1, n, 1\rangle \leq n$ and $\langle 1, 1, n\rangle \leq t$. By multiplying these inequalities we get $\langle n,n,n\rangle \leq t^{\otimes 3}$. Therefore, using \cref{prop:wecan}, we have $n^2 \leq \asympsubrank(\langle n,n,n\rangle) \leq \asympsubrank(t)^3$.
\end{proof}

Again, it is possible that $\omega = 2$, in which case it follows from \cref{prop:ub3} that $1 \leq \irr(t) \leq 2$ for balanced $t$.
More generally, it is possible that for every tensor $t \in \FF^{n \times n \times n}$ we have the highly non-trivial upper bound $\asymprank(t) \leq n$. In that case we have the following bounds.

\begin{proposition}\label{prop:ub4}
Let $t \in \FF^{n \times n \times n}$ be balanced and satisfy $\asymprank(t) \leq n$. Then
\[
1 \leq \irr(t) \leq 1.5.
\]
\end{proposition}
\begin{proof}
The claim follows directly from combining $\asympsubrank(t) \geq n^{2/3}$, which follows from balancedness as we saw in the proof of \cref{prop:ub3}, and the assumption $\asymprank(t) \leq n$.
\end{proof}

\begin{example}
The upper bound in \cref{prop:ub4} is tight. Namely, let $t = \cw_q$ be the small Coppersmith--Winograd tensor with parameter $q$. We will see from \cref{cweasy} that $\irr(t) \rightarrow 1.5$ when $q$ goes to infinity.
\end{example}

\section{Irreversibility implies barriers}\label{sec:barriers}

With the new notion of irreversibility available, we present a barrier for approaches to upper bound $\omega$ via an intermediate tensor $t$.
As we have discussed before, all recent successful upper bounds on $\omega$ have been obtained with constructions via an intermediate tensor. The results in this section tell us which tensors not to use as intermediate tensors and what necessary quality a good intermediate tensor has.

\subsection{The irreversibility barrier}\label{subsec:irr}
For any tensor $t$ the inequality
\begin{equation}
\relexp(\langle 2\rangle, t) \relexp(t, \langle2,2,2\rangle) \geq \omega
\end{equation}
holds
by the triangle inequality. Any such approach to upper bound $\omega$ respects the following barrier in terms of the irreversibility $\irr(t)$ of $t$.
\begin{theorem}\label{th1} For any tensor $t$ it holds that
\begin{equation}
\relexp(\langle 2\rangle, t) \relexp(t, \langle 2,2,2\rangle) \geq 2 \irr(t).
\end{equation}
\end{theorem}
\begin{proof} By the triangle inequality (\cref{basic}),
\begin{equation}
\relexp(\langle 2\rangle, t) \relexp(t, \langle 2,2,2\rangle) \relexp(\langle2,2,2\rangle, \langle2\rangle) \geq \relexp(\langle 2\rangle, t) \relexp(t, \langle 2\rangle) = \irr(t).
\end{equation}
Therefore, using the fact $\omega(\langle 2,2,2\rangle, \langle2\rangle) = \tfrac12$ from \eqref{subexp}, we have
\begin{equation}
\relexp(\langle 2\rangle, t) \relexp(t, \langle 2,2,2\rangle) \geq \frac{\irr(t)}{ \relexp(\langle2,2,2\rangle, \langle2\rangle)} = 2 \irr(t).
\end{equation}
This proves the claim.
\end{proof}

\cref{th1}, in particular, implies that if $\irr(t) > 1$, then $\relexp(\langle 2\rangle, t) \relexp(t, \langle2,2,2\rangle) > 2$. In other words,~we cannot prove $\omega = 2$ via an irreversible intermediate tensor. However, it is possible that there exists a \emph{sequence} of irreversible intermediate tensors with irreversibility converging to 1 that can be used to prove $\omega = 2$.

To conclude, the barrier just introduced describes what quality makes a tensor a good intermediate tensor, or put differently what intermediate tensors definitely not to use when we want to prove good upper bounds on $\omega$. Of course the barrier does not tell us explicitly which intermediate tensor to pick, but it provides a strong \emph{heuristic} of what to look for, namely low asymptotic rank and high asymptotic subrank. 

\subsection{Better barriers when the method has more structure} 
The barrier discussed in \cref{subsec:irr} does not tell the full story, and we will now discuss the natural and more subtle continuation. Namely, not only will we lose the game when the intermediate tensor is irreversible, we lose \emph{even more} when we use this intermediate tensor in a catalytic fashion, meaning that our algorithm is obtained from transforming a diagonal tensor to the tensor product of a diagonal tensor and a matrix multiplication tensor (via the intermediate tensor). We will explain this more, but for now it is important to know that this strategy is commonly used in the literature (with great success). However, as we will see in this section, such a catalytic approach boosts the previous barrier even more. Thus this section gives rise to a more precise heuristic which says that, while looking for intermediate tensors with small irreversibility, we may at the same time want to think about intermediate tensors that require less use of catalysis.

Catalysis in the general context of tensors is the phenomenon that for tensors $s$, $t$ and $u$, the inequality~$s \geq t$ may be false, while the inequality $s \otimes u \geq t \otimes u$ may be true. The latter inequality is called \emph{catalytic} with the tensor $u$ acting as a \emph{catalyst}. 

Catalysis is widely used in matrix multiplication algorithms albeit not under this explicit terminology.
We will now discuss more quantitatively what we mean when we say that an approach to upper bound $\omega$ uses catalysis.
Without loss of generality we may impose that the final step of any construction of a matrix multiplication algorithm is an application of the Schönhage $\tau$-theorem. 
The Schönhage~$\tau$-theorem (Strassen's general version \cite{strassen1988asymptotic}) says that
\begin{equation}
\asymprank \bigl( \bigoplus_{i=1}^q \langle a_i, b_i, c_i \rangle \bigr) \geq \sum_{i=1}^q (a_i b_i c_i)^{\omega/3}.
\end{equation}
In particular, for $a_i = b_i = c_i = a$, it holds that
\begin{equation}
\asymprank \bigl( \langle q\rangle \otimes \langle a, a, a \rangle \bigr) \geq q a^{\omega}.
\end{equation}
(And this inequality is in fact tight, although we will not go into that here.)
We are using here that the direct sum of $q$ copies of a tensor $t$ is isomorphic to the tensor product of the diagonal tensor $\langle q\rangle$ and $t$.
Equivalently, in the language of rates, for any~$\alpha, \beta \in \NN$ it holds that
\begin{equation}
\relexp(\langle2\rangle, \langle 2\rangle^\alpha \langle 2,2,2\rangle^\beta) \geq \alpha + \beta \omega
\end{equation}
that is
\begin{equation}
\frac{\relexp(\langle2\rangle, \langle 2\rangle^\alpha \langle 2,2,2\rangle^\beta) - \alpha}{\beta} \geq \omega.
\end{equation}
(Here $\alpha$ corresponds to $\log_2 q$ and $\beta$ corresponds to $\log_2 a$. For simplicity and concreteness we will consider only integer $\alpha$ and $\beta$.)
Thus for any tensor $t$ and for any~$\alpha, \beta \in \NN$ it holds that
\begin{equation}\label{mt}
\frac{\relexp(\langle2\rangle, t) \relexp(t, \langle 2\rangle^\alpha \langle 2,2,2\rangle^\beta) - \alpha}{\beta} \geq \omega.
\end{equation}%
This inequality \eqref{mt} is a method for upper bounding $\omega$ and we say that it is \emph{catalytic} when~$\alpha > 0$, the catalyst being the tensor $\langle 2\rangle^\alpha$.  In fact, upper bounds coming from the approach in the Coppersmith--Winograd paper and its follow-ups take the form of this inequality for specific $t$,~$\alpha$ and $\beta$.
The following barrier in terms of $\alpha, \beta$ and the irreversibility $\irr(t)$ of~$t$ for any method of the form~\eqref{mt} says that catalysis boosts the irreversibility barrier.

\begin{theorem}\label{th2}
For any tensor $t$ and $\alpha, \beta \in \NN$ it holds that %
\begin{equation}
\frac{\relexp(\langle2\rangle, t) \relexp(t, \langle 2\rangle^\alpha \langle 2,2,2\rangle^\beta) - \alpha}{\beta}\geq 2 \irr(t) + \frac{\alpha}{\beta}\bigl( \irr(t) - 1\bigr) \geq 2 \irr(t).
\end{equation}
\end{theorem}
\begin{proof}
By the triangle inequality,
\begin{equation}
\relexp(\langle 2\rangle, t) \relexp(t, \langle2\rangle^\alpha \langle 2,2,2\rangle^\beta) \relexp(\langle2\rangle^\alpha \langle 2,2,2\rangle^\beta, \langle2\rangle) \geq \relexp(\langle 2\rangle, t) \relexp(t, \langle 2\rangle) = \irr(t).
\end{equation}
Therefore, 
\begin{equation}\label{theother}
\relexp(\langle 2\rangle, t) \relexp(t, \langle2\rangle^\alpha \langle 2,2,2\rangle^\beta) \geq \frac{\irr(t)}{ \relexp(\langle2\rangle^\alpha \langle2,2,2\rangle^\beta, \langle2\rangle)}.
\end{equation}
We can bound this by
\begin{equation}\label{knowntobeequality}
\frac{\irr(t)}{ \relexp(\langle2\rangle^\alpha \langle2,2,2\rangle^\beta, \langle2\rangle)} \geq (\alpha + 2\beta) \irr(t)
\end{equation}
since for any tensors $s,t,u$ the inequality $\relexp(s\otimes t, u)^{-1} \geq \relexp(s,u)^{-1} + \relexp(t,u)^{-1}$ holds.
(The inequality in \eqref{knowntobeequality} is actually known to be an equality, since the asymptotic subrank is known to be additive on any collection of tensors of the form $\{\oplus_i t^{\otimes {n_i}} : n_i \in \NN\}$ for any fixed tensor $t$. This for example follows from the asymptotic subrank being characterized by the asymptotic spectrum of tensors that we will discuss in the next section.)
Combining \eqref{theother} and \eqref{knowntobeequality}, subtracting $\alpha$, dividing by $\beta$ and using that $\irr(t) - 1 \geq 0$ (\cref{biggerthanone}) gives the barrier
\begin{equation}
\frac{\relexp(\langle2\rangle, t) \relexp(t, \langle 2\rangle^\alpha \langle 2,2,2\rangle^\beta) - \alpha}{\beta} \geq \frac{(\alpha + 2\beta) \irr(t) - \alpha}{\beta}  = 2 \irr(t) + \frac{\alpha}{\beta}( \irr(t) - 1) \geq 2 \irr(t). %
\end{equation}
This proves the claim.
\end{proof}

As a corollary of the above theorem, we present a barrier on any approach of the following form. The Schönhage $\tau$-theorem implies that for any $a,b,c \in \NN_{\geq 1}$ and any tensor $t$ it holds that
\begin{equation}
\frac{\relexp(\langle2\rangle,t) \relexp(t, \langle2\rangle^\alpha \langle a, b, c \rangle) - \alpha}{\tfrac13\log_2(a b c)} \geq \omega.
\end{equation}
Again, we think of this inequality as a method for upper bounding $\omega$.
To prove a barrier for this method we define a cyclic symmetrization of $t$.
For any tensor $t = (t_{ijk}) \in \FF^{n \times n \times n}$ we define the cyclic permutations $(1,2,3)\cdot t$ and $(1,2,3)^2\cdot t$ by cyclically permuting the indices $i$, $j$ and $k$. We define the \emph{cyclic symmetrization} $\cyc(t)$ by $\cyc(t) \coloneqq t \otimes ((1,2,3)\cdot t) \otimes ((1,2,3)^2\cdot t) \in \FF^{n^3 \times n^3 \times n^3}$.
For example, $\cyc(\langle a,b,c\rangle) = \langle a,b,c\rangle \otimes \langle c,a,b\rangle \otimes \langle b,c,a\rangle = \langle abc,abc,abc\rangle$.
We prove the following barrier in terms of $a, b, c$,~$\alpha$ and the irreversibility of $\cyc(t)$.

\begin{corollary}\label{barrierth}
 For any tensor $t$ and $\alpha \in \NN$ and $a,b,c\in \NN_{\geq1}$ it holds that
\begin{align}
\frac{\relexp(\langle2\rangle, t) \relexp(t, \langle 2\rangle^\alpha \langle a,b,c\rangle) - \alpha}{\tfrac13 \log_2(abc)} &\geq 2 \irr( \cyc(t)) + \frac{\alpha}{\tfrac13\log_2(abc)}(\irr(\cyc(t)) - 1)\\
&\geq 2\irr(\cyc(t)).
\end{align}
\end{corollary}

One verifies that $\irr(t) \geq \irr(\cyc(t))$.
If $t$ is cyclically symmetric, then $\cyc(t) = t^{\otimes 3}$ and we have the equality $\irr(t) = \irr(\cyc(t))$.

\begin{proof}
We first prove that 
\[
\relexp(\langle2\rangle, t) \geq \relexp(\langle2\rangle, \cyc(t)^{\tfrac13}).
\]
Suppose that $\langle 2 \rangle^m \geq t^n$. Then also $\langle2\rangle^m \geq ((1,2,3)\cdot t)^n$ and $\langle2\rangle^m \geq ((1,2,3)^2\cdot t)^n$. Multiplying these inequalities gives $\langle2\rangle^{3m} \geq \cyc(t)^n$. We conclude that $m/n \geq \omega(\langle 2\rangle, \cyc(t)^{1/3})$.
By a similar argument we find that
\begin{equation}
\relexp(t, \langle2\rangle^\alpha \langle a,b,c\rangle) 
\geq %
\relexp(\cyc(t)^{\tfrac13}, \langle2\rangle^\alpha \langle 2,2,2\rangle^{\tfrac13\log_2(abc)}).
\end{equation}
Note that we are using real powers of tensors here inside the relative exponent $\omega(\cdot, \cdot)$. This is justified by taking powers of the relevant tensors and taking a limit.
Using both inequalities and then applying \cref{th2} gives
\begin{align}
\frac{\relexp(\langle 2\rangle, t) \relexp(t, \langle2\rangle^\alpha \langle a,b,c\rangle) -\alpha}{\tfrac13\log_2(abc)} &\geq \frac{\relexp(\langle 2\rangle, \cyc(t)) \relexp(\cyc(t), \langle2\rangle^\alpha \langle 2,2,2\rangle^{\tfrac13\log_2(abc)}) - \alpha}{\tfrac13 \log_2(abc)}\\
&\geq 2 \irr(\cyc(t)).
\end{align}
This proves the statement of the theorem. 
\end{proof}

\subsection{Better barriers through monomial irreversibility}\label{monomialsec}
Finally, we impose as an extra constraint that the transformation from the intermediate tensor $t$ to the matrix multiplication tensor happens via monomial restriction (\cref{monomial}), that is,~we consider the approach
\begin{equation}
\relexp(\langle2\rangle, t) \monrelexp(t, \langle2,2,2\rangle) \geq \omega
\end{equation}
and the more structured approaches
\begin{equation}
\frac{\relexp(\langle2\rangle, t) \monrelexp(t, \langle 2\rangle^\alpha \langle 2,2,2\rangle^\beta) - \alpha}{\beta} \geq \omega
\end{equation}
and
\begin{equation}
\frac{\relexp(\langle2\rangle,t) \monrelexp(t, \langle2\rangle^\alpha \langle a, b, c \rangle) - \alpha}{\tfrac13\log_2(a b c)} \geq \omega.
\end{equation}
The proofs in the previous sections can be directly adapted to prove:
\begin{theorem} For any tensor $t$ it holds that
\begin{equation}
\relexp(\langle2\rangle, t) \monrelexp(t, \langle2,2,2\rangle) \geq 2\monirr(t).
\end{equation}
\end{theorem}

\begin{theorem}\label{month2}
For any tensor $t$ and $\alpha, \beta \in \NN$ it holds that%
\begin{equation}
\frac{\relexp(\langle2\rangle, t) \monrelexp(t, \langle 2\rangle^\alpha \langle 2,2,2\rangle^\beta) - \alpha}{\beta}\geq 2 \monirr(t) + \frac{\alpha}{\beta}\bigl( \monirr(t) - 1\bigr) \geq 2 \monirr(t).
\end{equation}
\end{theorem}

\begin{corollary}\label{monbarrierth}
 For any tensor $t$ and $\alpha \in \NN$ and $a,b,c\in \NN_{\geq1}$ it holds that
\begin{align}
\frac{\relexp(\langle2\rangle, t) \monrelexp(t, \langle 2\rangle^\alpha \langle a,b,c\rangle) - \alpha}{\tfrac13 \log_2(abc)} &\geq 2 \monirr( \cyc(t)) + \frac{\alpha}{\tfrac13\log_2(abc)}(\monirr(\cyc(t)) - 1)\\
&\geq 2\monirr(\cyc(t)).
\end{align}
\end{corollary}

\section{Methods for lower bounding irreversibility}\label{sec:methods}

We have seen how lower bounds on irreversibility imply barriers.
By definition of irreversibility, such lower bounds come from lower bounds on asymptotic rank and upper bounds on asymptotic subrank.
In this section we discuss lower bounding irreversibility from four points of view: the asymptotic spectrum of tensors, the support functionals, the quantum functionals and the asymptotic slice rank.
We will use the support functionals to compute explicit barriers in \cref{sec:ex}; the rest of this section serves as a survey and to provide comparison.

\subsection{The asymptotic spectrum of tensors}\label{subsec:aspec}
We begin with an abstract and clean approach to characterizing irreversibility, and will move to practical bounds in the following sections. 
Let~$S$ be a family of tensors that is closed under $\otimes$ and $\oplus$ and that contains $\langle 1\rangle$.
Strassen~\cite{strassen1988asymptotic} introduced the \emph{asymptotic spectrum of tensors}  $\Delta(S)$ as the set of all maps $F : S \to \RR_{\geq 0}$ that satisfy for any tensors $s$ and $t$ in $S$ that
\begin{itemize}
\item $s \leq t \Rightarrow F(s) \leq F(t)$
\item $F(s \otimes t) = F(s) F(t)$
\item $F(s \oplus t) = F(s) + F(t)$
\item $F(\langle 1\rangle) = 1$.
\end{itemize} 
Strassen proved in \cite{strassen1988asymptotic} that for any tensor $t \in S$ it holds that
$\asympsubrank(t) = \min_{F \in \aspec(S)} F(t)$
and 
$\asymprank(t) = \max_{F \in \aspec(S)} F(t)$.
From this we directly obtain the following concise characterization of irreversibility in terms of the asymptotic spectrum of tensors $\Delta(S)$.

\begin{proposition} Let $t$ be a tensor in $S$. Then
\begin{equation}
\irr(t) = \frac{\max_{F \in \aspec(S)} \log F(t)}{\min_{G \in \aspec(S)} \log G(t)}.
\end{equation}
\end{proposition}
This characterization is clean but does not lead to a practical way of computing, or even lower bounding, the irreversibility $\irr(t)$ for a general tensor $t$.
This is because our knowledge of $\aspec(S)$ is very limited for any general family of tensors~$S$. Namely, a priori we only know that for every $i \in [3]$ the function $t \mapsto \rank(t_i)$ is in~$\Delta(S)$, where $t_i$ is the flattening as described in \cref{balanced}. Thus we have the inequalities $\asympsubrank(t) \leq \min_i \rank(t_i)$ and $\asymprank(t) \geq \max_i \rank(t_i)$. In fact, $\max_i \rank(t_i)$ is the best lower bound on $\asymprank(t)$ that we know of. We will keep using this lower bound on the asymptotic rank in the coming, more practical, sections. There are, however, more powerful tools than the flattening ranks $\rank(t_i)$ to upper bound~$\asympsubrank(t)$, which we will discuss now.

\subsection{The support functionals}\label{subsec:support}
The support functionals of Strassen provide an extremely powerful method for upper bounding the asymptotic subrank of tensors (and hence lower bounding irreversibility). The definition of the support functionals is technical at first sight, but as a reward these compact objects provide us with the best upper bounds on the asymptotic subrank.
Before we define the support functionals, we need to define some notions for tensors and for probability distributions.
For tensors $s, t \in \FF^{n_1 \times n_2 \times n_3}$ we write $s \cong t$ and say that \emph{$s$ and~$t$ are isomorphic} if there are invertible linear maps $A_i$ such that $(A_1, A_2, A_3) \cdot s = t$. 
Recall that we denote by 
\[
\supp(s) \coloneqq \{(i_1, i_2, i_3) \in [n_1] \times [n_2] \times [n_3] : s_{i_1, i_2, i_3} \neq 0\}
\]
the \emph{support of $s$}.
For any probability distribution $Q$ on $[n]$ let $H(Q) = \sum_{i \in [n]} Q(i) \log_2 (1/Q(i))$ denote the \emph{Shannon entropy of $Q$}. For any probability distribution $P$ on $[n_1] \times [n_2] \times [n_3]$ let~$P_i$ be the \emph{marginal distribution on~$[n_i]$}. That is, the marginal distribution $P_1$ is defined by $P_1(i_1) \coloneqq \sum_{i_2 \in [n_2]} \sum_{i_3 \in [n_3]} P(i_1, i_2, i_3)$ and the other marginal distributions are defined similarly.
Let $\prob(\supp(s))$ denote the set of all probability distributions on $\supp(s)$. (Thus for any element $P \in \prob(\supp(s))$ we can talk about its marginal distributions $P_i$ and their entropies $H(P_i)$.)
For any tensor $t$ and any probability vector~$\theta \in \RR^3$ Strassen defined the \emph{support functional} $\zeta^\theta$ by setting
\[
\zeta^\theta(t) \coloneqq 2^{\rho^{\theta}(t)}
\]
where
\[
\rho^\theta(t) \coloneqq \min_{s \cong t} \max_{P \in \prob(\supp(s))} \sum_{i=1}^3 \theta_i H(P_i).
\]
(We note that we believe that the name \emph{support functional} derives from the general concept of \emph{support functions} in convex analysis, rather than the support of tensors.)
The minimization over all $s$ isomorphic to $t$ appearing in the definition of $\rho^\theta(t)$ is generally not well understood, but, fortunately, for the sake of upper bounding $\zeta^\theta(t)$ it suffices to find one good $s$ isomorphic to $t$. Note that the Shannon entropy $H$ is a concave function and thus the weighted marginal entropy $\sum_{i=1}^3 \theta_i H(P_i)$ is a concave function of $P$. Therefore, the analysis of the maximization over $P$, being a convex program over an explicitly given domain, is usually straightforward.

It is not hard to see that the three flattening ranks $\rank(t_i)$ are among the support functionals, namely when we set $\theta$ to $(1,0,0)$, $(0,1,0)$ or $(0,0,1)$. Thus the support functionals can be thought of as interpolations between the three flattening ranks.
Strassen proves in \cite{strassen1991degeneration} the fundamental property that 
\[
\asympsubrank(t) \leq \zeta^\theta(t).
\]
There are numerous examples where $\min_\theta \zeta^\theta(t) < \min_i \rank(t_i)$, of which we will see some later.  It is not known whether the equality $\asympsubrank(t) = \min_\theta \zeta^\theta(t)$ holds in general. Strassen proved that equality holds for the family of \emph{tight} tensors~\cite{strassen1991degeneration}.\footnote{A tensor $t$  is called \emph{tight} if for some choice of basis there are injective maps $\alpha_1, \alpha_2, \alpha_3$ such that for every $a\in \supp(t)$ it holds that $\alpha_1(a_1) + \alpha_2(a_2) + \alpha_3(a_3) = 0$. Tight tensors play an important role in the laser method for constructing matrix multiplication algorithms, as the laser method requires the outer structure of the intermediate tensor to be tight.}
We conclude that we have the following lower bound on the irreducibility in terms of the flattening ranks and the support functionals:

\begin{proposition}
Let $t$ be a tensor. Then
\begin{equation}
\irr(t) \geq \frac{\max_{i} \log_2 \rank(t_i)}{\min_{\theta} \rho^\theta(t)}. %
\end{equation}
\end{proposition}

We will use the method of the support functionals to lower bound the irreversibility of some explicit tensors in \cref{sec:ex}.

Since we are primarily interested in using the support functionals to upper bound the asymptotic subrank, it is worth to observe that by the von Neumann minimax theorem we may for any fixed tensor $t$ express $\min_\theta \rho^\theta(t)$ in the concise form
\[
\min_\theta \rho^\theta(t) = \min_{s \cong t} \max_{P \in \prob(\supp(s))} \min_{i \in [3]} H(P_i).
\]
Indeed, the function $(\theta, P) \mapsto \sum_i \theta_i H(P_i)$ is convex in $\theta$ and concave in $P$ so the minimax theorem allows us to swap the maximization over $P$ and the minimization over $\theta$. Moreover, $\min_\theta \sum_i \theta_i H(P_i)$ is clearly attained when $\theta$ is one of the vertices $(1,0,0)$, $(0,1,0)$, $(0,0,1)$.

To further familiarize ourselves with the definition of $\zeta^\theta$, we discuss a simple example. For this example let $t$ be the so-called \emph{W-tensor} $t = e_{1,2,2} + e_{2,1,2} + e_{2,2,1}$. We will upper bound $\zeta^\theta(t)$ for $\theta = (1/3, 1/3, 1/3)$. In the evaluation of $\rho^\theta(t)$ we take $s$ to be equal to $t$. (This turns out to be optimal in this case, since $t$ is tight~\cite{strassen1991degeneration}.) Then the support of $s$ is the set
\[
\supp(s) = \{ (1,2,2), (2,1,2), (2,2,1)\}.
\]
Let $P \in \prob(\supp(s))$ assign probability $a$ to $(1,2,2)$, probability $b$ to $(2,1,2)$ and probability~$c$ to~$(2,2,1)$. 
The function $P \mapsto \sum_i \tfrac13H(P_i)$ is concave and invariant under permuting the values of~$a$,~$b$ and $c$. Thus, the maximum $\max_{P\in \prob(\supp(s))} \sum_i \tfrac13H(P_i)$ is attained by the symmetric probability distribution $P$ that assigns $a = b = c = 1/3$ to each element in $\supp(s)$. Then each marginal distribution $P_i$ assigns probability $a$ to $1$ and probability $2a$ to $2$. Thus $\max_{P\in \prob(\supp(s))} \sum_i \tfrac13H(P_i)$ equals the Shannon entropy of $P_i$, which is the \emph{binary entropy} $h(1/3) = -\tfrac13\log_2\tfrac13 -\tfrac23\log_2\tfrac23= 0.918...$. We conclude that for the W-tensor $t$ it holds that the support functional $\zeta^{(1/3,1/3,1/3)}(t)$ is upper bounded by $\zeta^{(1/3,1/3,1/3)}(t) \leq 2^{h(1/3)} = 1.88...$. Note that $1.88...$ is strictly smaller than the flattening rank $\rank(t_i) = 2$. We remark that in fact in this case we know the asymptotic monomial subrank to be $\monasympsubrank(t) = \zeta^{(1/3,1/3,1/3)}(t) = 2^{h(1/3)} = 1.88...$ since this $t$ is tight. Since a discussion of tightness is not in the scope of this paper we refer to~\cite{strassen1991degeneration} for that.

\section{Explicit barriers}\label{sec:ex}
We exhibit explicit barriers by computing lower bounds on the irreversibility of well-known intermediate tensors that play a crucial role in the best upper bounds on the matrix multiplication exponent $\omega$. These tensors are the small and big Coppersmith--Winograd tensors. Then we discuss the reduced polynomial multiplication tensors. Finally we discuss monomial irreversibility of structure tensors of finite group algebras and relations to the group-theoretic approach.

\subsection{Irreversibility of Coppersmith--Winograd tensors}\label{cwirr}
We now compute lower bounds for the irreversibility of the Coppersmith--Winograd tensors.
As mentioned, we will use the support functionals of Strassen \cite{strassen1991degeneration} in our computation to upper bound the asymptotic subrank. 

(Upper bounds on the asymptotic subrank of complex tensors may be obtained, not only from the Strassen support functionals, but also from the quantum functionals. For the tensors in \cref{cweasy} and \cref{cwhard}, however, it is known that the quantum functionals will give the same bound as the support functionals, since these tensors are \emph{free tensors} \cite[Section~4.3]{christandl2017universal}.)

\begin{theorem}[{Small Coppersmith--Winograd tensors \cite[Section 6]{MR1056627}}]\label{cweasy}
For any integer $q\geq 2$, the irreversibility of the small Copper\-smith--Winograd tensor
\begin{equation}
\cw_q \coloneqq \sum_{i=1}^q e_{0,i,i} + e_{i,0,i} + e_{i,i,0} 
\end{equation}
is lower bounded by
\begin{equation}\label{th9b}
2\irr(\cw_q) \geq 2\cdot \frac{\log_2(q+1)}{\log_2 3 - \tfrac23 + \tfrac 23 \log_2 q}.
\end{equation}
\end{theorem}

\begin{proof}
The rank of each flattening of $\cw_q$ equals $q+1$. Therefore, $\asymprank(\cw_q) \geq q+1$.
To upper bound the asymptotic subrank $\asympsubrank(\cw_q)$ one can upper bound the Strassen support functional with~$\theta = (1/3, 1/3, 1/3)$ as in \cite[Example 4.22]{christandl2017universal} by
\begin{equation}
\rho^\theta(\cw_q) \leq %
\log_2 3 - \frac{2}{3} + \frac{2}{3}\log_2 q.
\end{equation}
We find that 
\begin{equation}\label{smallcwb}
\irr(\cw_q) \geq \frac{\log_2(q+1)}{\log_2 3 - \tfrac23 + \tfrac 23 \log_2 q}.
\end{equation}
This proves the theorem.
\end{proof}

\begin{remark}
If $q > 2$, then the right-hand side of \eqref{th9b} is at least~$2.02..$ See the table in \cref{intro} for more values.
If $q=2$, however, then the right-hand side of \eqref{th9b} equals 2. 
\cref{cweasy} thus does not rule out using $\cw_2$ to prove that $\omega = 2$. Indeed, as observed in \cite[Section 11]{MR1056627}), if~$\omega(\langle2\rangle, \cw_2) = \log_2 3$, then $\omega = 2$. 

Currently, the best upper bound we have on $\relexp(\langle2\rangle, \cw_q)$ is $\log_2(q+2)$. If $\omega(\langle 2\rangle, \cw_q) = \log_2(q+2)$, then instead of \eqref{th9b} we get the better barrier
\begin{equation}\label{better}
2 \irr(\cw_q) \geq \frac{2\log_2(q+2)}{\log_2 3 - \tfrac23 + \tfrac 23 \log_2 q}.
\end{equation}
The right-hand side of \eqref{better} has a minimum value of
\begin{equation}
\frac{18}{5 \log_2 3} = 2.27..
\end{equation}
attained at $q = 6$.
\end{remark}

\begin{remark}
The following computation serves as a sanity check for our barrier given by \cref{th2}. Namely we see in an example how by putting some extra assumption the barrier given by \cref{th2} becomes tight.
Coppersmith and Winograd in \cite{MR1056627} used $\cw_q$ as an intermediate tensor in combination with the laser method and a certain \emph{outer structure}. Here outer structure refers to the idea of viewing the intermediate tensor as a block tensor. The outer structure is the block-support of the intermediate tensor, see also \cite[Section 9]{blaser2013fast}. For $\cw_q$ the outer structure is the W-tensor that we discussed in \cref{subsec:support} which has asymptotic subrank~$2^{h(1/3)}$.
When we impose that we apply the laser method on $\cw_q$ with this outer structure to get an upper bound~$\hat{\omega} \geq \omega$ we get the following better barrier via \cref{th2} or \cref{barrierth} and $\alpha = h(1/3)$ and $\beta = \tfrac13 \log_2(q)$:
\begin{equation}\label{cwlaser}
\hat{\omega} \geq 2 \irr(\cw_q) + \frac{h(1/3)}{\tfrac13\log_2(q)} (\irr(\cw_q) - 1).
\end{equation}
Some values of \eqref{cwlaser} are as follows, using \eqref{smallcwb} to get a bound on $\irr(\cw_q)$:
\begin{equation*}
\begin{minipage}{0.3\textwidth}
%
%
\begin{tabular}{ll}
\toprule
$q$ &  \\
\midrule
2 & 2\\
3 & 2.04744\\
4 & 2.10545\\
5 & 2.15338\\
6 & 2.19236\\
7 & 2.22455\\
\bottomrule
\end{tabular}
\end{minipage}
\end{equation*}
If in addition we assume that $\omega(\langle2\rangle, \cw_q) = \log_2(q+2)$, then we obtain the barrier
\begin{equation}\label{cwbetter}
\hat{\omega} \geq 2 \irr(\cw_q) + \frac{h(1/3)}{\tfrac13\log_2(q)} (\irr(\cw_q) - 1).
\end{equation}
Some values of \eqref{cwbetter} are as follows, using \eqref{better} to get a bound on $\irr(\cw_q)$:
\begin{equation*}
\begin{minipage}{0.3\textwidth}
%
%
\begin{tabular}{ll}
\toprule
$q$ &  \\
\midrule
2 & 3.24511\\
3 & 2.65678\\
4 & 2.50000\\
5 & 2.44072\\
6 & 2.41594\\
7 & 2.40614\\
8 & 2.40363\\
9 & 2.40492\\
10& 2.40824\\
11& 2.41266\\
\bottomrule
\end{tabular}
\end{minipage}
\end{equation*}
with minimum value of 2.40... These barriers in fact match the upper bound
\begin{equation}
\omega \leq \log_q \frac{4 (q+2)^3}{27}
\end{equation}
that was obtained by Coppersmith and Winograd by applying the laser method in the way described above. Thus our sanity check succeeds.
Other intermediate tensors with a given outer structure may be analyzed similarly.
\end{remark}

\begin{theorem}[{Big Coppersmith--Winograd tensors \cite[Section 7]{MR1056627}}]\label{cwhard}
For any integer $q \geq 1$ the irreversibility of the big Coppersmith--Winograd tensor
\begin{equation}
\CW_q \coloneqq e_{0,0,q+1} + e_{0,q+1,0} + e_{q+1,0,0} + \sum_{i=1}^q e_{0,i,i} + e_{i,0,i} + e_{i,i,0}
\end{equation}
is lower bounded by
\begin{equation}\label{CWb}
2\irr(\CW_q) \geq 
\begin{dcases}
\frac{2\log_2(3)}{f_1(\tfrac{1}{18}(\sqrt{33} - 3))} = 2.16.. & \textnormal{if } q = 1\\[0.5em]
\frac{2\log_2(4)}{f_2(\tfrac19)} = 2.17.. & \textnormal{if }q=2\\[0.5em]
\frac{2\log_2(q+2)}{f_q\bigl( \frac{3q - \sqrt{32+q^2}}{6(q^2-4)}\bigr)} & \textnormal{if }q \geq 3\\
\end{dcases}
\end{equation}
where 
\begin{equation}
f_q(x) \coloneqq - \Bigl(\frac23 - qx\Bigr)\log_2\Bigl(\frac23 - qx\Bigr) - q\cdot 2x \log_2(2x) - \Bigl(\frac13 - qx\Bigr)\log_2\Bigl(\frac13 - qx\Bigr).
\end{equation}
\end{theorem}

\begin{proof}
First, we compute the asymptotic rank of $\CW_q$.
One verifies directly that the matrix rank of each flattening of~$\CW_q$ equals $q+2$, so $q+2 \leq \asymprank(\CW_q)$. On the other hand, the well-known border rank upper bound $\borderrank(\CW_q) \leq q+2$ implies that $\asymprank(\CW_q) \leq q+2$. We conclude that $\asymprank(\CW_q) = q+2$.

Second, we will upper bound the asymptotic subrank $\asympsubrank(\CW_q)$ via the Strassen support functional~$\zeta^\theta(\CW_q)$ with $\theta = (1/3, 1/3, 1/3)$.
Recall that $\zeta^\theta(\CW_q) = 2^{\rho^{\theta}(\CW_q)}$ where
\begin{gather*}
\rho^\theta(\CW_q) = \min_{s \cong \CW_q} \max_{P \in \prob(\supp(s))} \sum_{i=1}^3 \theta_i H(P_i).
\end{gather*}
We let $s = \CW_q$. Then the support of $\CW_q$ is given by
\begin{equation}
\supp(\CW_q) = \{(0, i, i), (i, 0, i), (i, i, 0) : i \in [q]\} \cup \{(0,0,q+1), (0, q+1, 0), (q+1, 0, 0)\}.
\end{equation}
We will now evaluate $\max_{P \in \prob(\supp(\CW_q))} \sum_{i=1}^3 \theta_i H(P_i)$. We first observe that $\supp(\CW_q)$ is symmetric under permutation of the three coordinates and under permutation of the elements of~$[q]$. The function $P \mapsto \sum_{i=1}^3 \theta_i H(P_i)$ is concave. Thus by a simple convexity argument we may in our maximization, without loss of generality, consider only distributions $P$ that assign some probability $x$ to each of $(0, i, i)$, $(i,0,i)$ and $(0, i, i)$ for all $i \in [q]$, and probability~$\tfrac13 - qx$ to each of $(0, 0, q+1)$, $(0, q+1, 0)$ and $(q+1, 0, 0)$. Under this simplification, we have that
the average marginal entropy $\sum_{i=1}^3 \theta_i H(P_i)$ equals~$f_q(x)$ as defined in the theorem statement. It remains to compute the maximum of~$f_q(x)$ over all $0\leq x$ for which $3qx \leq 1$. From a calculus computation (easy to verify using computer algebra software) it follows that the maximum of $f_q(x)$ over all $0\leq x$ for which $3qx \leq 1$ is attained at 
\begin{equation}
x_0 = \begin{dcases}
\tfrac{1}{18}(\sqrt{33} - 3) & \textnormal{if } q = 1\\
\tfrac19 & \textnormal{if }q = 2\\
\frac{3q - \sqrt{32+q^2}}{6(q^2-4)} & \textnormal{if }q \geq 3.
\end{dcases}
\end{equation}
Then $\asympsubrank(\CW_q) \leq 2^{f_q(x_0)}$. Thus we have 
\[
2 \irr(\CW_q) \geq 2 \log_2 \asymprank(\CW_q) / \log_2 \asympsubrank(\CW_q) \geq 2 \log_2(q+2) / f_q(x_0),
\]
which evaluates to the bound in the claim.
\end{proof}

\begin{remark}
The lowest value of the right-hand side of \eqref{CWb} is $2.16..$ attained at $q = 1$. See the table in \cref{intro} for more values and see \cref{code} for code to compute any values.
\end{remark}

\subsection{Irreversibility of reduced polynomial multiplication tensors}\label{rpm}
The reduced polynomial multiplication tensors $t_n$ are an example of a natural family of tensors in which each tensor is irreversible, but where the irreversibility converges to 1 when $n$ goes to infinity.
The irreversibility of the tensors $t_n$ can be computed directly from the computation of the asymptotic rank and asymptotic subrank of $t_n$ in \cite[Theorem~6.7]{strassen1991degeneration}, which uses the support functionals. Namely, Strassen shows that
\[
\asymprank(t_n) = n
\]
and
\[
\asympsubrank(t_n) = z(n)
\]
where
\[
z(n) = \frac{g^n - 1}{g-1} g^{-2(n-1)/3}
\]
and $g > 1$ is the unique positive real solution to the equation
\[
\frac{1}{g-1} - \frac{n}{g^n-1} = \frac{n-1}{3}.
\]
Thus 
\[
\irr(t_n) = \frac{\log_2(n)}{\log_2 z(n)}.
\]
The limit $\lim_{n \to \infty} z(n) / n$ equals a constant, namely $0.84143...$ (see, e.g., \cite[Equation~4.11]{MR3631613}). It follows that 
\[
\lim_{n\to\infty} \irr(t_n) = \lim_{n\to\infty} \log_2(n) / \log_2 z(n) = 1.
\]
Thus $t_n$ is ``reversible in the limit''. In \cref{code} we provide code to compute the values of $\irr(t_n)$ for any $n$.

\subsection{Monomial irreversibility of structure tensors of finite group algebras}
We now discuss irreversibility and monomial irreversibility in the context of the group-theoretic approach developed in \cite{cohn2003group}. This approach produces upper bounds on~$\omega$ via intermediate tensors that are structure tensors of complex group algebras of finite groups. 
As defined earlier, let $\langle G\rangle_\FF$ denote the structure tensor of the group algebra $\FF[G]$ of the finite group $G$, in the standard basis, that is,
\begin{equation}
\langle G \rangle_\FF \coloneqq \sum_{g,h \in G} e_{g,h,gh} \in \FF^{G \times G \times G}.
\end{equation}
The group-theoretic approach~(in particular \cite[Theorem 4.1]{cohn2003group}) produces an inequality of the form
\begin{equation}
\langle G\rangle_\CC \mongeq \langle a,b,c\rangle
\end{equation}
(the base field that is used in \cite{cohn2003group} is $\CC$, in order to make the representation theory work)
which ultimately (see \cite[Eq.~(1)]{cohn2003group}) leads to the bound
\begin{equation}\label{gta}
\frac{\relexp(\langle2\rangle, \langle G\rangle_\CC)\monrelexp(\langle G\rangle_\CC, \langle a,b,c\rangle)}{\tfrac13\log_2 (abc)}
 \geq \omega
\end{equation}
where~$\mongeq$ and~$\monrelexp$ are the monomial restriction and monomial relative exponent defined in \cref{monomial}. 

Now the monomial irreversibility barrier from \cref{monomialsec} comes into play.
Upper bounds on the monomial asymptotic subrank of $\langle G\rangle_\FF$ have (using different terminology) been obtained in~\cite{MR3631613, blasiak2017groups, sawin2017bounds}. Those upper bounds imply that $\langle G\rangle_\FF$ is monomially irreversible for every nontrivial finite group $G$.
Together with our results in \cref{monomialsec} and the fact that the tensor~$\langle G\rangle_\FF$ is cyclically symmetric up to a permutation of the basis of one of the tensor legs, this directly leads to nontrivial barriers for the left-hand side of~\eqref{gta} for any fixed nontrivial group~$G$,
thus putting the work of \cite{MR3631613, blasiak2017groups,sawin2017bounds} in a broader context.
We have not tried to numerically optimize the monomial irreversibility barriers for group algebras. We gave an example of a monomial barrier for $G = \ZZ/3\ZZ$ in \cref{ste}.

It should be stressed again that although these results rule out using monomial restrictions from powers of any one single group, part of the hope of the group-theoretic approach is to use a family of groups (such as the symmetric group $S_n$ for increasing $n$) which is \emph{not} just powers of a single starting group. For families of abelian groups of bounded exponent (even if they are not powers of a single group), \cite{MR3631613} rules out STPP constructions reaching $\omega=2$, and for certain families of nilpotent groups (again, even if they are not just powers of a single group) \cite{blasiak2017groups} does similarly.

The results of \cite{MR3631613}, \cite{blasiak2017groups} and \cite{sawin2017bounds}, in fact, show slightly more than the aforementioned barriers for monomial restriction. While, over arbitrary fields, they only rule out monomial restrictions, over fields of bad characteristic (e.g.,~characteristic $p$ when $G$ is a $p$-group satisfying the relevant conditions of their theorems) they also rule out \emph{arbitrary degenerations}. This is because they show slice rank upper bounds, the slice rank of a diagonal tensor equals its rank, and having slice rank at most $r$ is a Zariski-closed condition (proved in \cite{sawin}).

Finally, we mention that the irreversibility barrier (rather than the monomial irreversibility barrier) does not rule out obtaining~$\omega = 2$ via~$\langle G\rangle_\CC$.
Namely, $\langle G\rangle_\CC$ is isomorphic to a direct sum of matrix multiplication tensors, $\langle G\rangle_\CC \cong \bigoplus_i \langle d_i, d_i, d_i\rangle$ and, therefore, the irreversibility satisfies $\irr(\langle G\rangle_\CC) = (\log_2 \sum_i d_i^\omega)/(\log_2 \sum_i d_i^2)$. %
Thus, if~$\omega = 2$, then~$\langle G\rangle_\CC$ is reversible.

\section{Outlook for further barriers}
In \cref{sec:ex} we used the support functionals of \cref{sec:methods} to upper bound the asymptotic subrank and thus lower bound the irreversibility of intermediate tensors. There are two other approaches to do this that we will discuss now. These approaches are at least as powerful as the support functionals for upper bounding asymptotic subrank, and we expect that there are examples where they perform better.
\subsection{The quantum functionals}\label{subsec:quantum}
For tensors over the complex numbers (i.e.,~when~$\FF = \CC$) we have a much deeper understanding of the theory of upper bounds on the asymptotic subrank. For any probability vector $\theta \in \RR^3$ the \emph{quantum functional} $F^\theta$ introduced in \cite{christandl2017universal}, or rather its logarithm, is defined as
\[
\log_2 F^\theta(t) = \max_{P \in \Pi(t)} \sum_{i=1}^r \theta_i H(P_i)
\]
where $\Pi(t)$ is the \emph{moment polytope of $t$}. It is outside the scope of this paper to go into the definition of the moment polytope. For that and a thorough discussion of all the connections to representation theory, invariant theory and quantum information theory we refer to \cite{christandl2017universal}. Here we will discuss the properties of the quantum functionals that are important for upper bounding the asymptotic subrank. The first crucial fact proven in \cite{christandl2017universal} is that the quantum functionals are in the asymptotic spectrum of all complex tensors $\Delta(\{\textnormal{tensors over $\CC$}\})$, which we discussed in \cref{subsec:aspec}.
We conclude that we can lower bound irreversibility in terms of the flattening ranks and the quantum functionals as follows:

\begin{proposition}
Let $t$ be a tensor over the complex numbers. Then
\begin{equation}
\irr(t) \geq \frac{\max_{i} \log_2 \rank(t_i)}{\min_\theta \log_2 F^\theta(t)}.
\end{equation}
\end{proposition}

Comparing to \cref{subsec:support}, it is proved in \cite{christandl2017universal} that the quantum functionals are at least as powerful as the support functionals when it comes to upper bounding the asymptotic subrank. Namely, for any tensor $t$ over the complex numbers holds that $F^\theta(t) \leq \zeta^\theta(t)$. It is an open problem whether this inequality can be strict. We note that, as opposed to the support functionals, the quantum functionals are defined as convex programs. It is an open problem whether the quantum functionals are efficiently computable.

\subsection{Asymptotic slice rank}
We finish \cref{sec:ex} by discussing the asymptotic slice rank as a method to upper bound asymptotic subrank, and the relations to the support functionals and the quantum functionals.
The \emph{slice rank} (introduced by Tao~\cite{tao} in the context of the cap set problem) of a tensor $t$ is the smallest number $r$ such that $t$ can be written as a sum of $r$ slice rank one tensors. A \emph{slice rank one tensor} is a tensor for which there is an $i \in [3]$ such that the flattening $t_i$ has matrix rank one.
The asymptotic subrank, the slice rank and the support functionals are related in the following way:
\begin{equation}\label{supportknow}
\asympsubrank(t) \leq \limsup_n \slicerank(t^{\otimes n})^{1/n} \leq \min_\theta \zeta^\theta(t).
\end{equation}
(See \cite{christandl2017universal}.)
Thus, in an asymptotic fashion, the slice rank upper bounds the asymptotic subrank and hence lower bounds irreducibility. 
Any analysis of $\limsup_n \slicerank(t^{\otimes n})^{1/n}$ that we are aware of in the literature boils down to evaluating $\min_\theta \zeta^\theta(t)$. In particular, we are not aware of any example for which the right-most inequality in \eqref{supportknow} is strict.

In fact, for oblique\footnote{A tensor $t \in \FF^{n_1} \otimes \FF^{n_2} \otimes \FF^{n_3}$ is called \emph{oblique} if the support $\supp(t) \in [n_1] \times [n_2] \times [n_3]$ in some basis is an antichain in the product of the natural orders on the $[n_i]$. The matrix multiplication tensors $\langle a,b,c\rangle$ are examples of oblique tensors.} tensors the right-most inequality is an equality~\cite{sawin} (see also \cite[Section 4.6]{phdzuiddam}) and, as mentioned in \cref{subsec:support}, for tight
tensors both inequalities are equalities~\cite{strassen1991degeneration}.
Over the complex numbers, the quantum functionals play a special role in this comparison. As we mentioned in \cref{subsec:quantum}, the quantum functionals satisfy $F^\theta(t) \leq \zeta^\theta(t)$. It was shown in~\cite{christandl2017universal} that the minimum over the parameter $\theta$ equals the asymptotic slice rank. We thus have
\begin{equation}\label{quantumknow}
\asympsubrank(t) \leq \limsup_n \slicerank(t^{\otimes n})^{1/n} = \min_\theta F^\theta(t) \leq \min_\theta \zeta^\theta(t).
\end{equation}
For free\footnote{A tensor $t$ is called \emph{free} if in some basis any two different $a,b\in \supp(t)$ differ in at least two entries. Every tight tensor is oblique and every oblique tensor is free.} 
tensors the right-most inequality in \eqref{quantumknow} is an equality~\cite{christandl2017universal}.

\appendix
\section{Code to verify the numerical examples}\label{code}
The following Mathematica code generates the barrier values in the tables in \cref{intro} to arbitrary precision and for arbitrary parameters $q$ and $n$.
\subsection{Small Coppersmith--Winograd tensor $\cw_q$}
\begin{verbatim}
In[1]:= Table[
  2 Log2[q + 1]/(Log2[3] - 2/3 + (2/3) Log2[q]), {q, 2, 7}] // N

Out[1]= {2., 2.02538, 2.06244, 2.09627, 2.12549, 2.15064}
\end{verbatim}
\subsection{Big Coppersmith--Winograd tensor $\CW_q$}
\begin{verbatim}
In[1]:= g[b_, q_] := -((2 b q Log[2 b])/
   Log[2]) - ((1 - 3 b q) Log[1/3 (1 - 3 b q)])/(
  3 Log[2]) + ((-b q - 2/3 (1 - 3 b q)) Log[b q + 2/3 (1 - 3 b q)])/
  Log[2]

In[2]:= f[q_] := 
 Piecewise[{{g[
     q/(2 (-4 + q^2)) + 1/6 Sqrt[(32 + q^2)/(-4 + q^2)^2] /. {q -> 1},
      1], q == 1}, {g[1/9, q], 
    q == 2}, {g[-((-3 q + Sqrt[32 + q^2])/(6 (-4 + q^2))), q], 
    q >= 3}}]

In[3]:= Table[2 Log2[q + 2]/f[q], {q, 1, 6}] // N

Out[3]= {2.16805, 2.17795, 2.19146, 2.20551, 2.21913, 2.23201}
\end{verbatim}

\subsection{Reduced polynomial multiplication $t_n$}
\begin{verbatim}
In[1]:= g[n_] := g[n] = g /. FindRoot[1/(g-1)-n/(g^n-1)==(n-1)/3,{g,1+2/n}]

In[2]:= z[n_] := (g[n]^n - 1)/(g[n] - 1) g[n]^(-2 (n - 1)/3)

In[3]:= Table[2 Log2[n]/Log2[z[n]], {n, 2, 7}]

Out[3]= {2.17795, 2.16805, 2.15949, 2.15237, 2.14641, 2.14135}
\end{verbatim}

\subsection*{Acknowledgements}
MC acknowledges financial support from the European Research Council (ERC Grant Agreement No.~337603 and 81876) and VILLUM FONDEN via the QMATH Centre of Excellence (Grant No.~10059). 
This research was supported by the National Research, Development and Innovation Fund of Hungary within the Quantum Technology National Excellence Program (Project Nr.~2017-1.2.1-NKP-2017-00001) and via the research grants K124152, KH129601 (PV).
This material is based upon work directly supported by the National Science Foundation Grant No.~DMS-1638352 and indirectly supported by the National Science Foundation Grant No.~CCF-1900460. Any opinions, findings and conclusions or recommendations expressed in this material are those of the authors and do not necessarily reflect the views of the National Science Foundation~(JZ). 

\raggedright
\bibliographystyle{alphaurlpp}
\bibliography{diss_allall}
\end{document}